\documentclass[conference,a4paper]{IEEEtran}

\usepackage[utf8]{inputenc} 
\usepackage[T1]{fontenc}
\usepackage{url}              % provides \url{...}
\usepackage{cite}             % improves presentation of citations

\usepackage[cmex10]{amsmath}  % Use the [cmex10] option to ensure complicance
\usepackage{mleftright}       % fix to wrong spacing of \left-,
\mleftright                   % \middle- \right-commands 

\usepackage{graphicx}         % provides \includegraphics{...} to
\usepackage{booktabs}         % fixes poor spacing in tables and
\usepackage{amssymb,amsfonts,amsthm,bm,bbm,dsfont,tabularx}
\newtheorem{theorem}{Theorem}

\newtheorem{claim}{Claim}
\newtheorem{proposition}{Proposition}

\newtheorem{lemma}{Lemma}
\newtheorem{definition}{Definition}

\usepackage{soul}

\DeclareSymbolFont{bbold}{U}{bbold}{m}{n}
\DeclareSymbolFontAlphabet{\mathbbold}{bbold}

\newcommand\blfootnote[1]{%
	\begingroup
	\renewcommand\thefootnote{}\footnote{#1}%
	\addtocounter{footnote}{-1}%
	\endgroup
}

\usepackage{color}

\newcommand{\cI}{\mathcal{I}}

\newcommand{\sI}{\mathcal{I}}
\newcommand{\sS}{\mathcal{S}}
\DeclareMathOperator*{\argmax}{arg\,max}
\DeclareMathOperator*{\argmin}{arg\,min}

\begin{document}

\title{Perturbation-Resilient Trades for Dynamic Service Balancing}

%%%%%%
\author{%
  \IEEEauthorblockN{%Anonymous Authors}
  \textbf{Jin Sima}, \textbf{Chao Pan}  and \textbf{Olgica Milenkovic}}
   \IEEEauthorblockA{Department of Electrical and Computer Engineering, University of Illinois Urbana-Champaign, USA \\\texttt{\{jsima,chaopan2,milenkov\}@illinois.edu}
   }
}

\maketitle

%%%%%
%% Abstract: 
%% If your paper is eligible for the student paper award, please add
%% the comment "THIS PAPER IS ELIGIBLE FOR THE STUDENT PAPER
%% AWARD." as a first line in the abstract. 
%% For the final version of the accepted paper, please do not forget
%% to remove this comment!
%%
\begin{abstract} 
A combinatorial trade is a pair of sets of blocks of elements that can be exchanged while preserving relevant subset intersection constraints. The class of balanced and swap-robust minimal trades was proposed in~\cite{pan2022balanced} for exchanging blocks of data chunks stored on distributed storage systems in an access- and load-balanced manner. More precisely, data chunks in the trades of interest are labeled by popularity ranks and the blocks are required to have both balanced overall popularity and stability properties with respect to swaps in chunk popularities. The original construction of such trades relied on computer search and paired balanced sets obtained through iterative combining of smaller sets that have provable stability guarantees. To reduce the substantial gap between the results of prior approaches and the known theoretical lower bound, we present new analytical upper and lower bounds on the minimal disbalance of blocks introduced by limited-magnitude popularity ranking swaps. Our constructive and near-optimal approach relies on pairs of graphs whose vertices are two balanced sets with edges/arcs that capture the balance and potential balance changes induced by limited-magnitude popularity swaps. In particular, we show that if we start with carefully selected  balanced trades and limit the magnitude of rank swaps to one, the new upper and lower bound on the maximum block disbalance caused by a swap only differ by a factor of $1.07$. We also extend these results for larger popularity swap magnitudes.
\end{abstract}
\blfootnote{This work was supported by the NSF grant CCF 1816913. Part of the works were presented at the International Symposiums on Information Theory 2022 and 2023.
}

\section{Introduction}
\label{sec:introduction}

In distributed storage systems, service (access) control methods are used to balance out read requests to servers and their load distributions, and thereby prevent service time bottlenecks~\cite{dau2018maxminsum,liu2019distcache,chee2020access,anderson2018service,aktacs2021service}. One approach to achieving this goal is to allocate carefully chosen  combinations of files or data chunks to different servers in order to both ensure efficient content reconstruction/regeneration in the presence of disk failures and level the ``average popularity'' (and thereby the access frequency) of files stored on different servers. To perform this repair-enabling and balanced allocation, new families of Steiner systems and related combinatorial designs, termed \emph{MinMax} Steiner systems, were introduced in~\cite{dau2018maxminsum,colbourn2020popularity,colbourn2021egalitarian}. There, in addition to block-intersection constraints common to designs, the authors also considered labels of elements in a block, which were deemed to be representative of their popularity rankings. The task at hand was to identify designs whose blocks have near-uniform sums of labels, or equivalently, close-to-balanced average file popularities. 

In practice, data popularity changes with time, and it is costly to redistribute files across servers after each popularity ranking change, especially when the magnitude of the changes is small. It is thus important to have combinatorial designs for which the discrepancy in the average server popularity scores is tightly restricted even in the presence of  perturbations in data popularity. 

One class of building blocks of designs, known as \emph{combinatorial trades}~\cite{hedayat1990theory,hwang1986structure, billington2003combinatorial}, has been proposed for use in dynamically changing storage systems. A trade is a pair of disjoint sets of blocks where the sets have the same cardinality and all blocks contain the same number of elements from a ground set. Furthermore, any subset of a given size is required to appear the same number of times within the blocks of each of the two constituent sets of blocks. As the name suggests, trades allow for exchanging collections of blocks in one of the two sets that violate certain system constraints, with blocks in the other set, as the latter may not cause violations. 

As an example~\cite{khosrovshahi2009trades}, one possible trade over a set of six elements $\{{1,2,3,4,5,6\}}$, with  blocks of cardinality three, and constraint subset size two, equals
$$\{{\{1,2,3\},\{1,2,4\},\{1,5,6\},\{2,5,6\},\{3,4,5\},\{3,4,6\}\}},$$
$$\{{\{1,2,5\},\{1,2,6\},\{1,3,4\},\{2,3,4\},\{3,5,6\},\{4,5,6\}\}}.$$
Here, each element in $\{{1,2,3,4,5,6\}}$ can be viewed as a data chunk whose value corresponds to its popularity ranking (with $1$ being the most popular and $6$ being the least popular chunk).
If at some point in time each block in the first set is stored on a separate server, it may become desirable to trade the blocks in this set with blocks in the second set, since in the former, the three most popular files are stored together which may lead to server failure due to large access demands. Note also that the servers total popularity scores, equal to the sum of the labels of the elements in their corresponding blocks, differ widely: in the first set of blocks, the lowest score of a block equals $5$ while the highest score equals $13$. This ``disbalance'' is undesirable as it causes server access and load issues.

To address the latter issues, the work~\cite{pan2022balanced,sima2023perturbation} introduced a special family of balanced and popularity-swap-robust (minimal) trades, i.e., trades with equal block sums and small induced popularity discrepancies caused by limited-magnitude popularity rank swaps. The analyses exclusively focused on trades with parameter $(4t,2t,t-1)$, but the proposed approach also  applies to more general parameter settings. The authors showed that the above-described family of  trades can be constructed through a careful selection of ``defining sets'' used in the minimal trades construction from~\cite{hwang1986structure}, and that the total popularity disbalance under popularity swaps of magnitude one scales linearly with $t$. Identifying good "defining sets" is crucial since it strongly influences the popularity discrepancy caused by popularity swaps. However, the "defining sets" problem is analytically difficult and computationally prohibitive 
due to a large combinatorial search space. The largest instance that can be solved by exhaustive search on a desktop computer is of size $4t=24$ (corresponding to $24$ data chunks of different popularity scores). Furthermore, the theoretical lower and upper bounds on the popularity disbalance caused by magnitude-one popularity swaps~\cite{pan2022balanced,sima2023perturbation} are nonnegligible and do not extend to general dynamic popularity-changes.

The contributions of this work are three-fold. First, we significantly improve the lower and upper bounds on the optimal disbalance of the defining sets from~\cite{pan2022balanced,sima2023perturbation}, and thereby nearly close the gap for the case of popularity swaps of magnitude-one. The results are based on a recursive construction (Section~\ref{sec:upper_bound}) of the defining sets, resulting in a popularity disbalance of $\frac{8}{5}(t-\frac{1}{4})$; this is to be compared to the currently best known bound $2t+O(1)$~\cite{pan2022balanced}. Second, we establish a significantly tighter lower bound $\frac{3}{2}(t-\frac{2}{3})$ (Section~\ref{sec:lower_bound}) on the total set discrepancy when compared to the lower bound $\frac{2}{3}(t-\frac{1}{3})$ proved in~\cite{pan2022balanced}. Our proof techniques use new families of paired graphs that describe the disbalance and potential disbalance increase induced by popularity swaps. Third, we extend the above results for popularity swaps of magnitude larger than one and defining sets of arbitrary cardinality.

The paper is organized as follows. In Section~\ref{section:preliminary}, we formally introduce the problem, describe the relevant notation and review known results. Section~\ref{sec:upper_bound} contains a description of the approach used to construct the new upper bound, while Section~\ref{sec:upper_bound} contains the proofs establishing the new lower bound. Extensions of these results are presented in Sections~\ref{sec:greater_pi},\ref{section:pitwo},\ref{section:greater_defining}. %Concluding remarks are presented in Section~\ref{sec:conclusion}.

\section{Preliminaries}
\label{section:preliminary}
%Given the space restrictions, we omit the details regarding the construction of balanced and swap-robust trades, which is technically independent of the problem of finding the defining sets 
%and can be found in~\cite{pan2022balanced}. 
For completeness, we recall the formal definition of trades and revisit the construction of minimal trades described in~\cite{hwang1986structure}. 

Let $V=[v]=\{1,\ldots,v\}$ be a set of integers and let $\mathcal{P}_k=\{P:P\subset V,|P|=k\}$ be the set of subsets of $V$ of cardinality $k,$ for some integer $k\le v$. 
\begin{definition}
A $(v,k,t)$ trade is a pair of sets $\{{T^{(1)}, T^{(2)}\}}$,  $T^{(i)} \subseteq \mathcal{P}_k, i=1,2,$ such that $|T^{(1)}| = |T^{(2)}|$ and 
$T^{(1)} \cap T^{(2)} = \emptyset$, satisfying the following property. For any $B_t \in \mathcal{P}_t$, the number of blocks in $T^{(1)}$ and the number of blocks in $T^{(2)}$ that contain $B_t$ is the same.
%\begin{align}\label{eq:tprop}
%&\left \{ B_t : \exists B^{(1)} \in T^{(1)} \text{ s.t. } B_t \subseteq B^{(1)} \right \} \notag\\
%= &\left \{ B_t : \exists B^{(2)} \in T^{(2)} \text{ s.t. } 
%B_t \subseteq B^{(2)} \right \}.
%\end{align}
%If for any $B_t \subseteq P_t(v)$, there exists at most one block $B^{(1)} \in T^{(1)}$ such that $B_t \subseteq B^{(1)}$ and also at most one block $B^{(2)} \in T^{(2)}$ such that $B_t \subseteq B^{(2)}$, than the trade is called a Steiner trade. 
The volume of the trade equals the number of blocks in $T^{(1)}$, i.e., $|T^{(1)}|$, and for a trade, we require that $|T^{(1)}| = |T^{(2)}|$. A $(v,k,t)$ trade is \emph{minimal} if it has the smallest possible number of blocks. 
\end{definition}
For each $\ell \in \{1, 2, \ldots, n\}$, let 
$$P_{\ell} = \{(i_1, i_1 + 1) \cdots (i_{\ell}, i_{\ell} + 1) \, | \, \text{ s.t. } \forall j\le \ell,\, i_j=1 \bmod 2,\},$$
be a set of permutations over a set of $2n$ elements, i.e., a set of permutations from the symmetric group $\mathbb{S}_{2n};$ in words, for odd positive integers $i_j,$ $(i_1, i_1 + 1)(i_2, i_2 + 1) \ldots (i_{\ell}, i_{\ell} + 1)$ denotes the permutation that swaps $i_j$ with $i_j + 1,$  $j = 1, \ldots, \ell,$ and leaves all other elements fixed. Furthermore, let $P_0 = \{(e)\}$ be the singleton set containing the identity permutation $e$. Also, define 
$$\Delta_{2n} = \bigcup\limits_{\ell \text{ even}} P_{\ell} \quad \text{ and } \quad \bar{\Delta}_{2n} = \bigcup\limits_{\ell \text{ odd}} P_{\ell}.$$
The following result was established in~\cite{hwang1986structure}.
\begin{theorem}\label{thm:hwang}
For all $v \geq k+t+1$, there exists a $(v,k,t)$ trade of volume $2^t$, and this volume is the smallest (minimal) possible volume for the given choice of parameters. Furthermore, there exists a unique family of minimal trades into which every nonminimal trade can be decomposed.
\end{theorem}
\begin{IEEEproof}
We provide a simple proof for the first claim as it allows us to introduce the notion of defining sets. 

Let $S_1, S_2, \ldots, S_{2t + 3}$ be a collection of subsets of the set $V$ in which the $2t + 3$ sets are arranged into $t+1$ pairs  
$$(S_1, S_2);\; (S_3, S_4);\; \ldots; (S_{2t+1}, S_{2t+2}),$$ with the addition of one unpaired set $S_{2t+3}$. We henceforth refer to these sets as \emph{defining sets} and the paired sets as \emph{companions}. The defining sets have the following properties: 
\begin{enumerate}
\item $S_i \cap S_j = \emptyset$ for $i \neq j$;
\item $|S_{2i-1}| = |S_{2i}| \ge 1$, for $i = 1, \ldots, t+1$;
\item $\sum\limits_{i = 1}^{t+2} |S_{2i-1}| = k$.
\end{enumerate}
Clearly, $V=\cup_{i}^{2t+3}\, S_i$. 

Next, define $T = \{T^{(1)}, T^{(2)}\}$ where the constituent sets 
 $T^{(1)}$ and $ T^{(2)}$ equal
%$$T^{(1)} = \bigcup\limits_{\sigma \in \Delta_{2t+2}} (S_{\sigma(1)} \cup S_{\sigma(3)} \cup \ldots \cup S_{\sigma(2t+1)} \cup S_{2t+3}) \quad \text { and}$$
\begin{align}\label{eq:trades}
   &\{{ (S_{\sigma(1)} \cup S_{\sigma(3)} \cup \ldots \cup S_{\sigma(2t+1)} \cup S_{2t+3}): \sigma \in \Delta_{2t+2} \}},\nonumber\\
   &\{{(S_{\bar{\sigma}(1)} \cup S_{\bar{\sigma}(3)} \cup \ldots \cup S_{\bar{\sigma}(2t+1)} \cup S_{2t+3}): \bar{\sigma} \in \bar{\Delta}_{2t+2}\}},
\end{align}
respectively. Note that we used $(S_{\sigma(1)} \cup S_{\sigma(3)} \cup \ldots \cup S_{\sigma(2t+1)} \cup S_{2t+3})$ and $(S_{\bar{\sigma}(1)} \cup S_{\bar{\sigma}(3)} \cup \ldots \cup S_{\bar{\sigma}(2t+1)} \cup S_{2t+3})$ to denote blocks whose elements represent the unions of the corresponding indexed sets $S$. The two straightforward results below were stated but not formally proved in~\cite{hwang1986structure}; we provide the proofs for completeness.

\begin{claim}
The paired sets in $T$ are minimal with respect to their volume (size), which equals $2^t$~\cite{hwang1986structure}.
\end{claim}
\begin{IEEEproof}
Since $\Delta_{2t+2}$ represents the set of all permutations with an even number of transpositions, for $t$ even we have 
$$|T^{(1)}| = |\Delta_{2t+2}| = {t+1 \choose 0} + {t+1 \choose 2} + \ldots + {t+1 \choose t} = 2^t,$$
$$|T^{(2)}| = |\bar{\Delta}_{2t+2}| = {t+1 \choose 1} + {t+1 \choose 3} + \ldots + {t+1 \choose t+1} = 2^t.$$
The same results hold for odd $t$.
%$$|T^{(1)}| = |\Delta_{2t+2}| = {t+1 \choose 0} + {t+1 \choose 2} + \ldots + {t+1 \choose t+1} = 2^t,$$
%$$|T^{(2)}| = |\bar{\Delta}_{2t+2}| = {t+1 \choose 1} + {t+1 \choose 3} + \ldots + {t+1 \choose t} = 2^t.$$
\end{IEEEproof}

\begin{claim}
$T$ is a $(v,k,t)$-trade~\cite{hwang1986structure}.
\end{claim}

\begin{IEEEproof}
It is easy to show that every block has size $k$ based on Property (3) of the collection of sets $S_1, S_2, \ldots, S_{2t + 3}$. 

To show that every $t$-subset of $V$ is contained in the same number of blocks of $T^{(1)}$ and of $T^{(2)}$, let $U$ be a $t$-subset of the ground set $V$.% with a nonempty intersection with $\bigcup\limits_{i = 1}^{2t+3} S_i$. 

\textbf{Case 1:} There exists some $i \in \{1, \ldots, t+1\}$ such that $U \cap S_{2i-1} \neq \emptyset$ and $U \cap S_{2i} \neq \emptyset$, i.e., $U$ has a nonempty intersection with a pair of companion sets. Then, by the construction of the trade, we have that $U$ is not completely contained in any of the blocks of the trade (i.e., in any block of $T^{(1)}$ and any block of $T^{(2)}$).

\textbf{Case 2:} For each pair of companion sets $S_{2i-1}, S_{2i}$, $1 \le i \le t+1$, $U$ has a nonempty intersection with at most one of the sets. Without loss of generality, assume that $U \cap S_{2 i - 1} \neq \emptyset$ for indices $i \in \{{i_1, i_2, \ldots, i_{h_1}\}}$ and $U \cap S_{2 j} \neq \emptyset$ for indices $j \in \{{j_1, j_2, \ldots, j_{h_2}\}}$, where $\{{i_1, i_2, \ldots, i_{h_1}\}} \cap \{{j_1, j_2, \ldots, j_{h_2}\}} = \emptyset$. Clearly, $h_1+h_2 \leq t$. For simplicity, assume that $h_1$ and $t$ are even, and that $h_2$ is odd. Then, $U$ is contained in
\begin{align}
{t+1-h_1-h_2 \choose  1} &+ {t+1-h_1-h_2 \choose  3}+ \ldots \notag \\ &+ {t+1-h_1-h_2 \choose  t-h_1-h_2} \notag
\end{align}
blocks of $T^{(1)}$ and in 
\begin{align}
{t+1-h_1-h_2 \choose  0} &+{t+1-h_1-h_2 \choose  2} \ldots \notag \\ &+ {t+1-h_1-h_2 \choose  t+1-h_1-h_2} 
\end{align}
blocks of $T^{(2)}$. In both cases, the sum equals $\frac{1}{2}2^{t+1-h_1-h_2}=2^{t-h_1-h_2}$.\end{IEEEproof} This completes the proof of the theorem.
\end{IEEEproof}
For purposes of data access balancing, we are interested in balanced trades, where the notion of \emph{balance} is captured via the block-sum discrepancy~\cite{pan2022balanced} defined as follows.
\begin{definition}
Let $B=(b_1,b_2,\ldots,b_k)$ be a block in $T^{(i)}$, where $i \in \{{1,2\}}$. The block-sum of $B$ equals $\Sigma_B=\sum_{i=1}^k\,b_i$, while the minimum and maximum block sums of $T^{(i)}$ are defined as $\min_{B \in T^{(i)}}\, \Sigma_B$ and 
$\max_{B \in T^{(i)}}\, \Sigma_B$, for $i \in \{{1,2\}}$. The block-discrepancy of $T^{(i)}$ is defined as 
$\max_{B \in T^{(i)}}\, \Sigma_B - \min_{B \in T^{(i)}}\, \Sigma_B$, for $i\in\{{1,2\}}$.
\end{definition}
Henceforth, we tacitly assume that the integer-valued elements of the blocks correspond to the popularity  of data chunks and that no two popularities are the same: label ``$1$'' indicates the most popular data chunk, while label ``$v$'' indicates the least popular data chunk. 

For simplicity, throughout the paper we focus on minimal trades with parameters $(4t,2t,t-1)$. The first and obvious observation made in~\cite{pan2022balanced} is that one can construct balanced trades with (perfect) zero block discrepancy by forcing the companion sets $S_{2i-1}$ and $S_{2i}$, $i\in [t]$, in the construction~\eqref{eq:trades}  to have the same sum of elements, i.e. $\Sigma(S_{2i-1})=\Sigma(S_{2i})$, $i\in[t]$, where $\Sigma(S_{i'})=\sum_{j\in S_{i'}} j,$ for $i'\in[2t]$\footnote{We will not consider potential balanced trades for which the companion set sum constraint is not satisfied, since in this case the problem becomes hard to analyze.}. If the sum-constraint for the companion sets is not satisfied, we say that the companion (and corresponding defining) sets are imbalanced. The total imbalance is measured by the total \emph{discrepancy}, defined as
$\sum_{i=1}^{t}|\Sigma(S_{2i-1})-\Sigma(S_{2i})|.$ Balanced trades ensure that both the blocks that may violate some storage design criteria as well as those that can be used to replace them guarantee uniform loads/access to the servers.

Note that this access balance might be compromised in the presence of dynamic popularity changes. Here, we model these changes by considering their ``magnitude'' $p$, referring to swaps of labels or popularity values $\{(i_1,j_1),(i_2,j_2),\ldots,(i_m,j_m)\},$ such that $|i_\ell-j_\ell|\le p$ and $i_1,\ldots,i_m,j_1,\ldots,j_m$ are all distinct. As an example, for $p=1$, $t=2$, and $V=[8]$, the allowed popularity swap sets $I$  are subsets of $\{(1,2),(2,3),(3,4),(4,5),(5,6),(6,7),(7,8)\}$ with the additional constraint that the same element is not included in two different swaps (as swaps with a common element correspond to \emph{nonadjacent} swaps and hence larger popularity change magnitudes). For example, for $p=1$, $\{{(1,2),(3,4),(5,6)\}}$ is a valid set of popularity swaps, while $\{{(1,2),(2,3),(5,6),(7,8)\}}$ is not since the element $2$ appears in two swaps. We do not impose any restrictions on the cardinality of the set $I$, although for some derivations pertaining to $p>1$ we assume that $|i_{\ell}-j_{\ell}|=p$ rather than $|i_{\ell}-j_{\ell}| \leq p$ (Theorem \ref{thm:maingeneralp}). In the latter case, we will use the notation $=p$ for all relevant subscripts.

In the absence of popularity swaps, the total set discrepancy of balanced defining sets is equal to zero. The \emph{worst-case total balance discrepancy for a given defining set} is the largest total set discrepancy that any valid collection of swaps $I$ could possibly induce. The goal of our work is to find (near) optimal balanced defining sets with respect to the worst-case total set discrepancy, i.e., we seek balanced defining sets that have the smallest worst-case total set discrepancy. For example, when $t=2$, the two companion sets 
$$
S_1=\{1,8\},S_2=\{3,6\}; S_3=\{2,7\},S_4=\{4,5\},
$$
constitute an example of balanced defining sets which are also optimal, among balanced sets, with respect to the worst-case balance discrepancy. The claim that the trade is balanced is easy to verify since the sums of entries of the companion sets are equal. For the second claim, we have two swaps $I=\{{(1,2),(5,6)\}}$ of magnitude one that lead to the following changes in the companion sets 
$$
S_1^\prime=\{2,8\},S_2^\prime=\{3,5\}; S_3^\prime=\{1,7\},S_4^\prime=\{4,6\},
$$
and result in a total set discrepancy of $|(2+8)-(3+5)|+|(1+7)-(4+6)|=4$, which by computer search is the worst-case total set discrepancy for any choice of $I$ involving magnitude-one swaps and for the above defining sets. For other choices of balanced defining sets, say $$S_1=\{1,4\},S_2=\{2,3\};S_3=\{5,8\},S_4=\{6,7\},$$ 
the worst-case total set discrepancy with respect to the choice of $I$ is $6$ (which is clearly $>4$), making these defining sets suboptimal with respect to popularity change stability to swaps of magnitude one.

Formally, we denote the set of all allowed collections of swaps of magnitude $p$ for $4t$ elements by $\sI_{t,p}$, and the set of all balanced collections of defining sets by $\sS_t$. The goal is to find defining sets $(S_1^*,\ldots,S_{2t}^*)$ such that their worst-case total discrepancy $\max_{I_{t,p}\in\sI_{t,p}} D(S_1^*,\ldots,S_{2t}^*;I_{t,p};t)$ is the smallest possible, i.e.,
\begin{align}
&(S_1^*,\ldots,S_{2t}^*)=\argmin_{(S_1,\ldots,S_{2t})\in\mathcal{S}_t}\max_{I_{t,p}\in\mathcal{I}_{t,p}}D(S_1,\ldots,S_{2t};I_{t,p};t), \notag \\
&D^*(t,p)=\min_{(S_1,\ldots,S_{2t})\in\mathcal{S}_t}\max_{I_{t,p}\in\mathcal{I}_{t,p}}D(S_1,\ldots,S_{2t};I_{t,p};t), \notag \\
&D(S_1,\ldots,S_{2t};I_{t,p};t)=\sum^{t}_{i=1}\left|\Sigma(S_{2i-1}^\prime)-\Sigma(S_{2i}^\prime)\right|,
\end{align}
where $(S_1^\prime,\ldots,S_{2t}^\prime)$ denote the defining sets after the popularity swaps have been applied.

Note that it is impossible to find $(S_1^*,\ldots,S_{2t}^*)$ even for modest values of $t$ by exhaustive search, as the size of the search space grows super-exponentially in $t$. In~\cite{pan2022balanced}, both upper and lower bounds on the number of different partitions of $[4t]$ induced by valid 
balanced defining sets of cardinality $2$ were derived  using integer partition formulas. We provide improved bounds in the following theorem. Specifically, we consider
\begin{align*}
  &N=\Bigg|\Bigg\{ \Big\{ \{S_1,S_2\},\ldots,\{S_{2t-1},S_{2t}\} \Big\} :(S_1,\ldots,S_{2t})\in\sS_t\Bigg\}\Bigg|, \notag \\
  &|S_i|=2, \; i=1,2,\ldots,t.
\end{align*}
i.e., the number of ways to partition $[4t]$ into $t$ pairs of balanced companion sets of cardinality $2$. 
%In the following, we show that the bounds on $N$ are almost tight.
\begin{theorem}\label{thm:tightbounds}
The number of partitions $N$ is at least $\exp(t\log t-O(t))$ and at most $\exp(2t\log t+O(t))$, where we used the natural logarithm. Note that the total number of partitions (balanced and unbalanced) equals $\frac{(4t)!}{t!2^t}=\exp(3t\log t+O(t))$. 
\end{theorem}
\begin{proof}
%The arguments follow similarly as in \cite{pan2022balanced}. We briefly state the proof as follows. 
We first show that $N$ is at least $$\frac{(2t)!}{t!2^t}=\frac{(2t)!}{t!}\ge t!=\exp(t\log t-O(t)).$$
This follows by forcing the sums $\sum(S_i)$ to be equal to $4t+1$ for all $i\in[2t]$. Clearly, the sets $(S_1,\ldots,S_{2t})$ are balanced. Also, there are $(2t)!$ ordering choices for $(S_1,\ldots,S_{2t})$. Furthermore, swapping $S_{2i-1}$ and $S_{2i}$, $i\in [t]$ and swapping $\{S_{2i_1-1},S_{2i_1}\}$ and $\{S_{2i_2-1},S_{2i_2}\}$, $i_1,i_2\in[t]$ does not change the partition. Therefore, 
$$N \geq \frac{(2t)!}{t!2^t}\ge \exp(t\log t-O(t)).$$

We show next that $N$ is at most $$(8t)^{2t}\mathcal{P}(2t(4t+1))=\exp(2t\log t+O(t)),$$
where $\mathcal{P}(n)$ is the number of integer partitions of $n$, i.e., the number of ways to represent $n$ as a sum of positive integers, i.e., the number of sets $\{x_1,\ldots,x_m\}$ where $m\in [n]$ such that $\sum^m_{j=1}x_j=n$. A well-known result by Hardy and Ramanujan~\cite{hardy1918asymptotic} established the asymptotic growth of the partition function provided below
$$\mathcal{P}(n)\sim \frac{1}{4\sqrt{3}n}\exp\left(\pi\sqrt{\frac{2n}{3}}\right).$$

Note that $$\sum^{2t}_{i=1}\big(\sum(S_i)\big)=\sum^{4t}_{j=1}j=2t(4t+1).$$
Therefore, $\{\sum(S_1),\ldots,\sum(S_{2t})\}$ is an integer partition of $2t(4t+1)$. Note that for each value of $\sum(S_i)$, $i\in[2t]$, there are at most $\sum(S_i)\le 8t$ choices for the set $S_i$. Hence, for each choice  of $\{\sum(S_1),\ldots,\sum(S_{2t})\}$, there are at most $(8t)^{2t}$ choices for the partition $\{\{S_1,S_2\},\ldots,\{S_{2t-1},S_{2t}\}\}$. Invoking the Hardy-Ramanujan formula, we obtain that there are at most $\mathcal{P}(2t(4t+1))\sim \exp(O(t))$ possibilities for $\{\sum(S_1),\ldots,\sum(S_{2t})\}$. This leads to at most $(8t)^{2t}\exp(O(t))=\exp(2t\log t+O(t))$ choices for the partition $\{\{S_1,S_2\},\ldots,\{S_{2t-1},S_{2t}\}\}$ such that $(S_1,\ldots,S_{2t})\in \sS_t$. 

The above arguments are of interest because of their connections to the theory of integer partitions \cite{andrews2004integer}, on which there is a large body of works that allows for adding different problem constraints. Still, simpler arguments, like the ones described below, result in an upper bound with matching first order exponent and improved second order exponent. Suppose we choose the companion sets sequentially. When choosing the $i$th companion sets, there are at most $(4t-4i+3)(4t-4i+2)(4t-4i)\le (4t)^3$ options. Therefore, there are at most $(4t)^{3t}$ number of choices for the $t$ pairs of companion sets. On the other hand, the permutations of the $t$ pairs of companion sets do not change the partition. Hence, the total number of partition is at most $\frac{(4t)^{3t}}{t!}=\exp(2t\log t+O(t))$. 
\end{proof}
Next, we remark that a recursive construction for balanced defining sets robust under $p=1$ popularity swaps was discussed in~\cite{pan2022balanced}, where the building blocks of larger trades were optimal defining sets for small values of $t$ (i.e., $t\in\{1,2,3,4,5,6\}$) found via computer search. The work also presented a lower bound on the worst-case total discrepancy of the form $\frac{2}{3}(t-\frac{1}{3})$, again for $p=1$. 

In what follows, we first improve this lower bound to $\frac{3t-2}{2}$, and provide a completely analytic recursive construction that provably achieves a smaller worst-case total set discrepancy compared to the previously reported one, for the case $p=1$. We then generalize the lower bound to apply to arbitrary popularity changes of magnitude exactly $p$, and present a recursive construction of the companion sets for $p=2$ that has a worst-case total set discrepancy close to the lower bound. Our main results are summarized in the following theorems.
\begin{theorem}\label{thm:main}
For any integer $t>0$, we have
\begin{align}\label{eq:lower}
   D^*(t,1)\ge \frac{3t-2}{2}. 
\end{align}
Moreover, for any positive integer $z\ge 2$ and $t = 5 \cdot 2^{z-2}-1$, we have
\begin{align}\label{eq:upper}
   D^*(t,1)\le 2^{z+1}-2,
\end{align}
which implies $D^*(t,1)\le \frac{8t-2}{5}$. Consequently, the upper and lower bound only differ by a constant factor $1.07$.
\end{theorem}
The following result generalizes Theorem~\ref{thm:main} for arbitrary values of $p$, and uses a similar proof technique. We remark that the result is obtained by restricting to cases where the popularity swaps $(i,j)$ are of distances exactly $|i-j|=p$ and therefore, is not tight.
\begin{theorem}\label{thm:maingeneralp}
For any integer $t>0$, we have
\begin{align}\label{eq:lowergeneralp}
   D^*(t,p)\ge \frac{p[3t-2(p-1)]}{2}. 
\end{align}
\end{theorem}
Finally, the next theorem provides a (slight) improvement of the lower bound in Theorem~\ref{thm:maingeneralp} for $p=2$, as well as a an upper bound.
\begin{theorem}\label{thm:p=2}
For any integer $t>0$, we have
\begin{align}\label{eq:lowerp=2}
   D^*(t,2)\ge \frac{35t-40}{11}. 
\end{align}
For any integer $z\ge 1$ and $t=2z+1$, we have
\begin{align}\label{eq:upperp=2}
  D^*(t,2)\le 9z+4=\frac{9t-1}{2},
\end{align}
%Moreover, for any positive integer $z\ge 2$ and $t = 5 \cdot 2^{z-2}-1$, we have
%\begin{align}\label{eq:upper}
%   D^*(t,1)\le 2^{z+1}-2,
%\end{align}
%which implies $D^*(t,1)\le \frac{8t-2}{5}$. %The upper and lower bound only differ by a constant factor $1.07$.
\end{theorem}

The proofs of these results are presented in the next sections. %\textcolor{red}{are these for $p\leq 2$ or $p=2$?}

\section{A Recursive Construction for $p=1$}
\label{sec:upper_bound}
As stated in Theorem~\ref{thm:tightbounds}, the search space for the optimal partition of $[4t]$ among all valid partitions induced by balanced defining sets is intractable. In addition, the number of optimal partitions and the corresponding optimal set discrepancy depend on the value of $t$. Numerical results obtained using brute-force search for $p=1$ were presented in~\cite{pan2022balanced} and are summarized in Tab. \ref{tab:p=1}.
\begin{table}[h!]
\caption{Number of optimal partitions and the corresponding set discrepancy for $p=1$}
\label{tab:p=1}
\centering
\begin{tabular}{ |c|c|c|}
 \hline
 Value of $t$& Number of optimal partitions &Optimal set discrepancy\\
 \hline 
 3& 10 &6\\
 \hline
 4& 1 &6\\
 \hline
 5& 1 &8\\
 \hline
 6& 22 &10\\
 \hline
\end{tabular}
\end{table}
Note that for $t=4$ and $t=5$ the optimal constructions are \emph{unique}\footnote{We do not have any analytical proofs for the uniqueness results. The problem of determining for which values of $t$ the optimal defining sets are unique is open.}. 

In what follows, we provide constructions for defining sets $(S_1,\ldots,S_{2t})\in\mathcal{S}_t$ for $p=1$ and for values of $t$ that satisfy $t=5\cdot2^{z-2}-1$, $z\ge 2$, such that the total discrepancy is upper bounded by
$$
\max_{I_{t,1}\in\sI_{t,1}} D(S_1,\ldots,S_{2t};I_{t,1};t) \leq \frac{8t}{5}-\frac{2}{5}.
$$
To begin with, consider the case $z=2$ and $t=5\cdot 2^{z-2}-1=4$ and add superscripts to the defining sets $S$ to indicate the value of $z$. %Using exhaustive search we can show \cite{pan2022balanced} that 
We use 
the unique optimal defining sets for $t=4$ \cite{pan2022balanced}, which equal
\begin{align}\label{eq:comp16}
&S^2_1=\{1,16\},S^2_2=\{8,9\}; \, S^2_3=\{2,7\},S^2_4=\{4,5\};\\
&S^2_5=\{10,15\},S^2_6=\{12,13\}; \, S^2_7=\{3,14\},S^2_8=\{6,11\}.\nonumber
\end{align}
The discrepancy upon performing the worst-case magnitude-one swaps, e.g., $\{(1,2),(6,7),(11,12)\}$ (note that there are multiple sets of magnitude-one swaps resulting in the same discrepancy), equals $6=2^{z+1}-2$. %\textcolor{red}{spell out what the worst-case swaps are.} 
Next, we describe a recursive construction for the defining sets $(S^{z+1}_1,\ldots,S^{z+1}_{2t})$ and $t=5\cdot 2^{(z+1)-2}-1$, based on $(S^{z}_1,\ldots,S^{z}_{2t}),$ for $t=5\cdot 2^{z-2}-1$. The construction is as follows:
\begin{align}\label{eq:construction}
    &S^{z+1}_{i}=S^{z}_{i}+1, \forall i\in[5\cdot 2^{z-1}-1],\nonumber\\
    &S^{z+1}_{i}=S^{z}_{i-5\cdot 2^{z-1}+2}+5 \cdot 2^{z}-1,
    \forall i\in[5 \cdot 2^{z-1}-1,5 \cdot 2^{z}-4],\nonumber\\
    &S^{z+1}_{5 \cdot 2^{z}-3}=\{1,5 \cdot 2^{z+1}-4\},\nonumber\\
    &S^{z+1}_{5 \cdot 2^{z}-2}=\{5 \cdot 2^{z}-2,5 \cdot 2^{z}-1\},
\end{align}
where for an integer set $S$ and an integer $a$, we define $S+a=\{x+a:x\in S\}$. An example of the construction for $z=4$ is depicted in Fig.~\ref{fig:examplepartition}. The intuition behind the construction is as follows: once we fix the last companion sets to $S^{z+1}_{5 \cdot 2^{z}-3}=\{1,5\cdot2^{z+1}-4\}$ and $S^{z+1}_{5 \cdot 2^{z}-2}=\{5 \cdot 2^{z}-2,5 \cdot 2^{z}-1\}$, the elements used in the other $t-1$ %\textcolor{red}{why the subscript 1 for t? Not explained or properly introduced?} 
companion sets must come from two disjoint and symmetric intervals $[2,5\cdot 2^{z}-3]$ and $[5\cdot 2^{z},5 \cdot 2^{z+1}-5]$. Therefore, we can reuse the construction of companion sets $(S^{z}_1,\ldots,S^{z}_{2t})$ for $t=5 \cdot2^{z-2}-1$. %\textcolor{red}{why the subscript 2 for t? Not explained or properly introduced?}

\begin{figure}[t]
\centering
\includegraphics[width=1\linewidth]{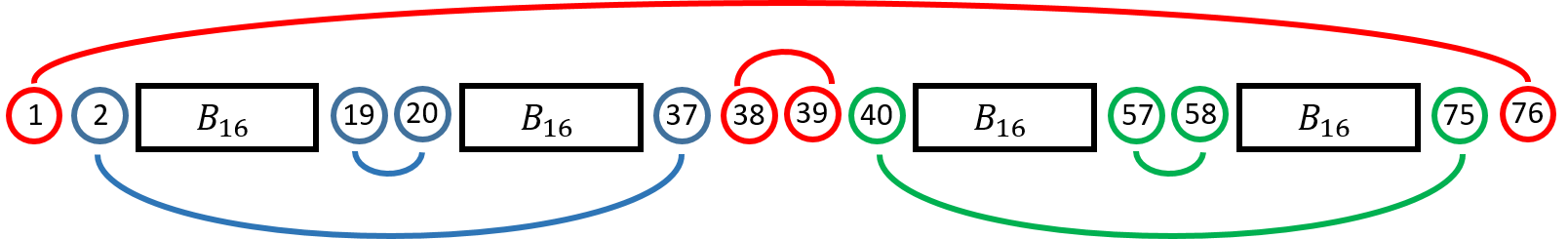}
\caption{An example of the recursive construction for $z=4$ and $t=19$. We use the companion sets in~\eqref{eq:comp16}, denoted by $B_{16}$, as building blocks. Integers within the same companion sets have the same color. The companion sets have cardinality $2$.}
\label{fig:examplepartition}
\end{figure}

Next, we upper-bound the worst-case total discrepancy for the recursively constructed collection of companion sets as follows. Let 
$$
d_z=\max_{I_{t,1}\in\mathcal{I}_{t,1}}D(S^z_1,\ldots,S^z_{2t};I_{t,1};t),\; t=5 \cdot 2^{z-2}-1.
$$ 
We have the following lemma.

\begin{lemma}\label{lemma:dzrecursion}
For $z\ge 2$, we have $d_{z+1}\le 2d_z+2$.
\end{lemma}

\begin{IEEEproof}
The proof relies on separately examining scenarios for the discrepancy change induced by an allowed set of swaps $I_{t,1}$. The cases include:
\begin{enumerate}
    \item $|\Sigma(S^{\prime z+1}_{5\cdot 2^{z}-3})-\Sigma(S^{\prime z+1}_{5 \cdot 2^{z}-2})|=0$;
    \item $|\Sigma(S^{\prime z+1}_{5 \cdot 2^{z}-3})-\Sigma(S^{\prime z+1}_{5 \cdot 2^{z}-2})|=1$;
    \item $|\Sigma(S^{\prime z+1}_{5 \cdot 2^{z}-3})-\Sigma(S^{\prime z+1}_{5 \cdot 2^{z}-2})|=2$.
\end{enumerate}

\textbf{Case 1):} In this case we have 
\begin{align*}
&|I_{t,1}\cap\{(1,2),(5 \cdot 2^{z}-3,5 \cdot 2^{z}-2)\}|\\
=&|I_{t,1}\cap\{(5 \cdot 2^{z}-1,5 \cdot 2^{z}),(5 \cdot 2^{z+1}-3,5 \cdot 2^{z+1}-4)\}|. 
\end{align*}

If $|I_{t,1}\cap\{(1,2),(5\cdot 2^{z}-3,5 \cdot 2^{z}-2)\}|=0$, then the discrepancy corresponding to the sets $(S^{\prime z+1}_{1},\ldots,S^{\prime z+1}_{5\cdot 2^{z-1}-2})$, as well as to the sets  $(S^{\prime z+1}_{5 \cdot 2^{z-1}-1},\ldots,S^{\prime z+1}_{5 \cdot 2^{z}-4}),$ is at most $d_z$. Thus, the total discrepancy is at most $2d_z$. 

If $|I_{t,1}\cap\{(1,2),(5 \cdot 2^{z}-3,5 \cdot 2^{z}-2)\}|=1$, then the discrepancy corresponding to the sets $(S^{\prime z+1}_{1},\ldots,S^{\prime z+1}_{5 \cdot 2^{z-1}-2})$, as well as to the sets $(S^{\prime z+1}_{5\cdot 2^{z-1}-1},\ldots,S^{\prime z+1}_{5 \cdot 2^{z}-4})$ is at most $d_z+1$. Hence, the total discrepancy is at most $2d_z+2$. 

Similarly, if $|I_{t,1}\cap\{(1,2),(5 \cdot 2^{z}-3,5 \cdot 2^{z}-2)\}|=2$, then the set discrepancies corresponding to $(S^{\prime z+1}_{1},\ldots,S^{\prime z+1}_{5\cdot 2^{z-1}-2})$ and to $(S^{\prime z+1}_{5 \cdot 2^{z-1}-1},\ldots,S^{\prime z+1}_{5 \cdot 2^{z}-4})$ are at most $d_z$.%$(S_{5\cdot 2^{z-1}-2}^{\prime z+1}$, 
%=\{2,5*2^{z}-3\},
%$S_{5 \cdot 2^{z-1}-1}^{\prime z+1})$
%=\{5*2^{z-1}-1,5*2^{z-1}\})$ 
%and $(S_{5 \cdot 2^{z}-5}^{\prime z+1}%=\%{5*2^{z},5*2^{z+1}-5\}
%,S_{5 \cdot 2^{z}-4}^{\prime z+1}%=\{5*2^{z}+5*2^{z-1}-3,5*2^{z}+5*2^{z-1}-2\}
%)$ are at most $1$. 
Then, 
the total discrepancy of $(S^{\prime z+1}_{1},\ldots,S^{\prime z+1}_{5 \cdot 2^{z}-2})$ is at most $2d_z$.

\textbf{Case 2):} For this case, exactly one of the two conditions holds:
\begin{align*}
&|I_{t,1}\cap\{(1,2),(5 \cdot 2^{z}-3,5 \cdot 2^{z}-2)\}|=1, \text{or}\\ 
&|I_{t,1}\cap\{(5 \cdot 2^{z}-1,5 \cdot 2^{z}),(5 \cdot 2^{z+1}-3,5 \cdot 2^{z+1}-4)\}|=1.
\end{align*}
By symmetry, we can assume that $|I_{t,1}\cap\{(1,2),(5 \cdot 2^{z}-3,5 \cdot 2^{z}-2)\}|=1$. Then exactly one of the following holds:
\begin{align*}
&|I_{t,1}\cap\{(5 \cdot 2^{z}-1,5 \cdot 2^{z}),(5 \cdot 2^{z+1}-3,5 \cdot 2^{z+1}-4)\}|=0, \\ 
&|I_{t,1}\cap\{(5 \cdot 2^{z}-1,5 \cdot 2^{z}),(5 \cdot 2^{z+1}-3,5 \cdot 2^{z+1}-4)\}|=2.
\end{align*}
Hence, the discrepancy induced by the sets $(S^{\prime z+1}_{1},\ldots,S^{\prime z+1}_{5 \cdot 2^{z-1}-2})$ is at most $d_z+1$ and the discrepancy induced by $(S^{\prime z+1}_{5 \cdot 2^{z-1}-1},\ldots,S^{\prime z+1}_{5 \cdot 2^{z}-4})$ is at most $d_z$. Hence, the total discrepancy for all sets $(S^{\prime z+1}_{1},\ldots,S^{\prime z+1}_{5 \cdot 2^{z}-2})$ is at most $d_z+1+d_z+1=2d_z+2,$ where the additional $1$ comes from the discrepancy of $(S^{\prime z+1}_{5\cdot 2^{z}-3},S^{\prime z+1}_{5\cdot 2^{z}-2})$.

\textbf{Case 3):} For this case, exactly one of the two conditions holds:
\begin{align*}
&|I_{t,1}\cap\{(1,2),(5 \cdot 2^{z}-3,5 \cdot 2^{z}-2)\}|=2, \text{ or }\\
&|I_{t,1}\cap\{(5 \cdot 2^{z}-1,5 \cdot 2^{z}),(5 \cdot 2^{z+1}-3,5 \cdot 2^{z+1}-4)\}|=2.
\end{align*}
Similarly to Case 2), by symmetry we can assume that $|I_{t,1}\cap\{(1,2),(5 \cdot 2^{z}-3,5 \cdot 2^{z}-2)\}|=2$ and thus $|I_{t,1}\cap\{(5 \cdot 2^{z}-1,5 \cdot 2^{z}),(5 \cdot 2^{z+1}-3,5 \cdot 2^{z+1}-4)\}|=0$. Then, the discrepancy induced by $(S^{\prime z+1}_{1},\ldots,S^{\prime z+1}_{5 \cdot 2^{z-1}-2})$ is at most $d_z$, while the discrepancy induced by $(S^{\prime z+1}_{5 \cdot 2^{z-1}-1},\ldots,S^{\prime z+1}_{5 \cdot 2^{z}-4})$ is at most $d_z$. Hence, the total discrepancy of all sets $(S^{\prime z+1}_{1},\ldots,S^{\prime z+1}_{5 \cdot 2^{z}-2})$ is at most $d_z+d_z+2=2d_z+2$, where the additional term equal to $2$ comes from the discrepancy of $(S^{\prime z+1}_{5\cdot 2^{z}-3},S^{\prime z+1}_{5\cdot 2^{z}-2})$.
\end{IEEEproof}
By Lemma~\ref{lemma:dzrecursion} and the fact that $d_2=6$, we arrive at
\begin{align}\label{eq:recur_ub}
d_z\le 2^{z+1}-2=\frac{8t}{5}-\frac{2}{5}.
\end{align}

\section{A Lower Bound on $D^*(t,1)$}
\label{sec:lower_bound}
We prove next the bound in~\eqref{eq:lower}, which is significantly more challenging to establish than the upper bound. Before proceeding with the proof, we first define an \emph{undirected graph} $G_{\text{swp}}(I_{t,1})$ that describes a collection of swaps $I_{t,1}\in\mathcal{I}_{t,1}$, where the nodes represent the companion sets and the edges represent the swaps in $I_{t,1}$. The number of edges in $G_{\text{swp}}(I_{t,1})$ is the number of swaps in $I_{t,1}$, which, as shown later, equals one half of the discrepancy when $I_{t,1}$ is chosen to maximize the discrepancy. To get a lower bound on $D^*(t,1)$, the idea is to obtain a lower bound on the number of edges in $G_{\text{swp}}(I_{t,1})$. To this end, we also define a \emph{directed graph} $G_{pot}(I_{t,1})$ with the same vertex set as $G_{\text{swp}}(I_{t,1})$ that represents the companion sets. Different from the undirected edges in $G_{\text{swp}}$, the directed arcs in $G_{pot}(I_{t,1})$ describe the potential discrepancy change for swaps \emph{not included in} $I_{t,1}$. %\textcolor{red}{what are the arcs of this graph?}. 
We observe that the number of edges incident to a node in $G_{\text{swp}}(I_{t,1})$ is at least the number of input arcs to the same node in $G_{\text{pot}}(I_{t,1})$, which follows from the observation that the swaps not included in $I_{t,1}$ cannot increase the highest possible discrepancy when $I_{t,1}$ is selected accordingly (i.e., so as to maximize the discrepancy). This gives a lower bound on the number of edges in $G_{\textup{swp}}(I_{t,1})$. In the following, we give formal definitions of the graphs $G_{\text{swp}}(I_{t,1})$ and $G_{\text{pot}}(I_{t,1})$. We use brackets $\{v_{i_1},v_{i_2}\}$ to denote edges in an undirected graph and parentheses $(v_{i_1},v_{i_2})$ to denote an arc in a directed graph. We note that the parentheses notation for denoting arcs is not to be confused with the parentheses notation for denoting swaps.  

Fix the defining sets $(S_1,\ldots,S_{2t})\in \mathcal{S}_t$. For any collection of allowed swaps $I_{t,1}\in\mathcal{I}_{t,1}$, define an unweighted, undirected graph $G_{\text{swp}}(I_{t,1})=(V_{\text{swp}}(I_{t,1}),E_{\text{swp}}(I_{t,1}))$ that describes the swap set $I_{t,1}$. Specifically, the node set equals $V_{\text{swp}}(I_{t,1})=\{v_i\}^t_{i=1},$ where node $v_i$ is indexed by the companion sets $(S_{2i-1},S_{2i})$. For any two nodes $v_{i_1},v_{i_2}\in V_{\text{swp}}(I_{t,1})$, there exists an edge $\{v_{i_1},v_{i_2}\}$ between $v_{i_1}$ and $v_{i_2}$ if there exists a swap $(i,i+1)\in I_{t,1}$ such that $i$ and $i+1$ are in the sets $S_{2i_1-1}\cup S_{2i_1}$ and $S_{2i_2-1}\cup S_{2i_2}$, respectively. Note that $i_1$ and $i_2$ can be the same, meaning that $G_{\text{swp}}(I_{t,1})$ is allowed to have self loops. In addition, multiple edges are also allowed between the same pair of nodes. An example graph $G_{\text{swp}}(I_{t,1})$ for $t=3$ and $I_{t,1}=\{(2,3),(8,9),(11,12)\}$ is shown in Figure~\ref{fig:graph_example}. Note that a similar, \emph{but different}, definition of a swap graph $G_{\text{swp}}(I_{t,1})$ was used in~\cite{pan2022balanced}.
%and we consider the multiplicity of an edge $\{v_{i_1},v_{i_2}\}\in E_{\text{swp}}(I)$ when counting the number of edges.
% i.e., $i\in S_{2i_1-1}\cup S_{2i_1}$, $i+1\in S_{2i_2-1}\cup S_{2i_2}$ or $i+1\in S_{2i_1-1}\cup S_{2i_1}$, $i\in S_{2i_2-1}\cup S_{2i_2}$.

\begin{figure}[t]
\centering
\includegraphics[width=0.9\linewidth]{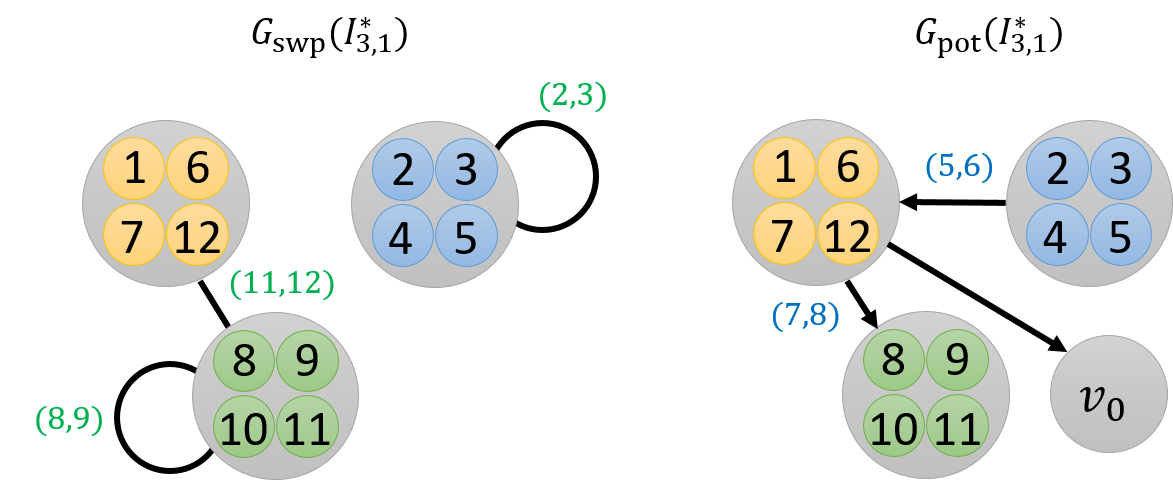}
\caption{Two graphs $G_{\text{\text{swp}}}(I_{3,1}^*)$ and $G_{\text{\text{pot}}}(I_{3,1}^*)$ defined for the worst-case popularity swaps $I_{3,1}^*=\{(2,3),(8,9),(11,12)\}$ (i.e., $p=1$, $t=3$) obtained by computer s. Each gray circle contains $4$ elements, which constitute a pair of companion sets, and represents the same node in both graphs. The three edges in $G_{\text{\text{swp}}}(I_{3,1}^*)$ correspond to the three swaps (marked in green in $G_{\text{\text{swp}}}(I_{3,1}^*)$) in $I_{3,1}^*$. The three arcs correspond to the remaining swaps (marked in blue in $G_{\text{\text{pot}}}(I_{3,1}^*)$) not included in $I_{3,1}^*$ that cause potential discrepancy increase for different pairs of companion sets.}%\textcolor{red}{Two connected components of $G_{\text{\text{swp}}}^1(I_{t,1}^*),G_{\text{\text{swp}}}^2(I_{t,1}^*)$ are circled in dashed-red in $G_{\text{\text{swp}}}(I_{t,1}^*)$ -why do we care and not clear; also, need explanation regarding edges/arcs}. \textcolor{red}{For analytical convenience, the compositions of connected components are always $G_{\text{\text{swp}}}^1(I_{t,1}^*),G_{\text{\text{swp}}}^2(I_{t,1}^*)$, no matter if the underlying graph is $G_{\text{\text{swp}}}(I_{t,1}^*)$ or $G_{\text{\text{pot}}}(I_{t,1}^*)$. Do not understand this at all?}}
\label{fig:graph_example}
\end{figure}

In addition to $G_{\text{swp}}(I_{t,1})$, we also define a directed graph $G_{\text{pot}}(I_{t,1})=(V_{\text{pot}}(I_{t,1}),E_{\text{pot}}(I_{t,1}))$ that describes the \text{potential} discrepancy change induced by swaps $(i,i+1)$ that are not included in $I_{t,1}$. The node set $V_{\text{pot}}(I_{t,1})=V_{\text{swp}}(I_{t,1})\cup\{v_0\}$ equals the node set $V_{\text{swp}}(I_{t,1}),$ but is augmented by an auxiliary node $v_0$ that acts as a ``virtual'' companion set and is discussed later. To define the directed edge set (i.e., arc set) $E_{\text{pot}}(I_{t,1})$, we consider all potential swaps $(i,i+1)\notin I_{t,1}$ for $i\in [4t-1]$. Note that the potential swap $(i,i+1)\notin I_{t,1}$ and the selected swaps in $I_{t,1}$ are allowed to share an element. For a swap $(i,i+1)\notin I_{t,1}$, an arc  $(v_{i_1},v_{i_2})$ directed from $v_{i_1}$ to $v_{i_2}$, $i_1,i_2\in[t]$, exists if  
$(i,i+1)$ satisfies one of the following: 
\begin{enumerate}
    \item $i\in S_{2i_1}$, $i+1\in (S_{2i_2-1}\cup S_{2i_2})\backslash S_{2i_1}$, and $\Sigma(S_{2i_1-1}^\prime)<\Sigma(S_{2i_1}^\prime)$ after performing the swaps in $I_{t,1}$.
    \item $i+1\in S_{2i_1}$, $i\in (S_{2i_2-1}\cup S_{2i_2})\backslash S_{2i_1}$, and $\Sigma(S_{2i_1-1}^\prime)>\Sigma(S_{2i_1}^\prime)$ after performing the swaps in $I_{t,1}$.
    \item $i\in S_{2i_1-1}$, $i+1\in (S_{2i_2-1}\cup S_{2i_2})\backslash S_{2i_1-1}$, and $\Sigma(S_{2i_1-1}^\prime)>\Sigma(S_{2i_1}^\prime)$ after performing the swaps in $I_{t,1}$.
    \item $i+1\in S_{2i_1-1}$, $i\in (S_{2i_2-1}\cup S_{2i_2})\backslash S_{2i_1-1}$, and $\Sigma(S_{2i_1-1}^\prime)<\Sigma(S_{2i_1}^\prime)$ after performing the swaps in $I_{t,1}$.
    \item $i\in S_{2i_1-1}$, $i+1\in (S_{2i_2-1}\cup S_{2i_2})\backslash S_{2i_1-1}$, and $\Sigma(S_{2i_1-1}^\prime)=\Sigma(S_{2i_1}^\prime)$ after performing the swaps in $I_{t,1}$.
    \item $i+1\in S_{2i_1}$, $i\in (S_{2i_2-1}\cup S_{2i_2})\backslash S_{2i_1}$, and $\Sigma(S_{2i_1-1}^\prime)=\Sigma(S_{2i_1}^\prime)$ after performing the swaps in $I_{t,1}$.
\end{enumerate}
Intuitively, $(v_{i_1},v_{i_2})\in E_{\text{pot}}(I_{t,1})$ if $i$ and $i+1$ are in $S_{2i_1-1}\cup S_{2i_1}$ and $S_{2i_2-1}\cup S_{2i_2}$, respectively, and the swap $(i,i+1)$ increases the discrepancy $|\Sigma(S_{2i_1-1}^\prime)-\Sigma(S_{2i_1}^\prime)|$ of the companion sets $(S_{2i_1-1}^\prime,S_{2i_1}^\prime)$, whenever the discrepancy of these two sets is $>0$. In addition, if the discrepancy of the companion sets $(S_{2i_1-1}^\prime,S_{2i_1}^\prime)$ equals $0$, we have $(v_{i_1},v_{i_2})\in E_{\text{pot}}(I_{t,1})$ only when the swap $(i,i+1)$ leads to a positive set difference $\Sigma(S_{2i_1-1}^\prime)-\Sigma(S_{2i_1}^\prime)$. Note that conditions $5)$ and $6)$ are introduced to ensure that swaps corresponding to multiple arcs can simultaneously increase the set discrepancy of the same pair of companion sets. Again, we allow $i_1=i_2$, i.e., we allow self-loops in $G_{\text{pot}}(I_{t,1})$. Moreover, multiple arcs $(v_{i_1},v_{i_2})\in E_{\text{pot}}$ are also allowed. 

In addition to the arcs among nodes $\{v_i\}^t_{i=1}$ in $G_{\text{pot}}(I_{t,1})$, we add an arc $(v_{i_1},v_0)\in E_{\text{pot}}(I_{t,1})$ if exactly one of the following conditions holds:
\begin{enumerate}
    \item $1\in S_{2i_1-1}\cup S_{2i_1}$ and the swap $(0,1)$ increases the discrepancy of the companion sets $(S_{2i_1-1}^\prime,S_{2i_1}^\prime),$ when $\Sigma(S_{2i_1-1}^\prime)\ne\Sigma(S_{2i_1}^\prime)$;
    \item $1\in S_{2i_1},$ when $\Sigma(S_{2i_1-1}^\prime)=\Sigma( S_{2i_1}^\prime)$.
\end{enumerate}
Note that the swap $(0,1)$ is not allowed in $I_{t,1}$ and is only included in $E_{\text{pot}}(I_{t,1})$ for the purposes of simplifying the analysis. Similarly, we allow an arc $(v_{i_2},v_0)\in E_{\text{pot}}(I_{t,1})$ for the swap $(4t,4t+1)$ provided that:
\begin{enumerate}
    \item $4t\in S_{2i_1-1}\cup S_{2i_1}$ and the swap $(4t,4t+1)$ increases the discrepancy of the companion sets $(S_{2i_1-1}^\prime,S_{2i_1}^\prime),$ when $\Sigma(S_{2i_1-1}^\prime)\ne\Sigma(S_{2i_1}^\prime)$;
    \item $4t\in S_{2i_1},$ when $\Sigma(S_{2i_1-1}^\prime)=\Sigma( S_{2i_1}^\prime)$.
\end{enumerate}
An example of the graph $G_{\text{pot}}(I_{t,1})$ for $t=3$ and $I_{t,1}=\{(2,3),(8,9),(11,12)\}$ is shown in Figure~\ref{fig:graph_example}.

Given the defining sets $(S_1,$ $\ldots, S_{2t})$, let $I_{t,1}^*$ be a swap-set of the smallest size among all swap-sets that lead to the worst-case total discrepancy, i.e.,
$$
I_{t,1}^*\in \argmax_{I_{t,1}\in\mathcal{I}_{t,1}} D(S_1,\ldots,S_{2t};I_{t,1};t).
$$
We have the following lemma.
\begin{lemma}\label{lem:optimalswap}
Each pair $(i,i+1)\in I_{t,1}^*$ increases the total discrepancy by $2$, i.e., 
\begin{align}\label{eq:discrepancy}
    D(S_1,\ldots, S_{2t};I_{t,1}^*;t)=2|I_{t,1}^*|.
\end{align} 
\end{lemma}
\begin{proof}
Note that a swap $(i,i+1)$ either increases or decreases the discrepancy of two companion sets $(S_{2i_1-1},S_{2i_1})$ and $(S_{2i_2-1},S_{2i_2})$ that contain $i$ and $i+1$, respectively, by $1$, unless $i$ and $i+1$ are in the same set $S_{2i_1-1}$ or $S_{2i_1}$, which is a case that can be ignored. Therefore, the contribution of the swap $(i,i+1)\in I_{t,1}^*$ to the total discrepancy is $2$, $0$, or $-2$. Since $I_{t,1}^*$ is of smallest size, each swap in $I_{t,1}^*$ contributes $2$ to the total set discrepancy. Hence,~\eqref{eq:discrepancy} holds. 
\end{proof}
%Since our subsequent analysis does not depend on the value of $t$, we henceforth omit the subscript $t$ from the notation. 
For simplicity, throughout the rest of this section, we use $I^*$ to denote $I^*_{t,1}$. We also write $G_{\text{swp}}=G_{\text{swp}}(I^*)=(V_{\text{swp}},E_{\text{swp}})$, and $G_{\text{pot}}=G_{\text{pot}}(I^*)=(V_{\text{pot}},E_{\text{pot}})$. 

The graph $G_{\text{swp}}$ can be partitioned into a set of connected components $G^j_{\text{swp}}=(V^j_{\text{swp}},E^j_{\text{swp}})$, $j\in[J]$, as illustrated by the example in Figure~\ref{fig:graph_example}.
For $v_i\in V_{\text{swp}}\cup\{v_0\}$, let 
$$
d^{\text{pot}}_{\text{in}}(v_i)=|\{(v_j,v_i):(v_j,v_i)\in E_{\text{pot}},v_j\in V_{\text{swp}}\}|
$$
be the in-degree of node $v_i$ and 
$$
d^{\text{pot}}_{\text{out}}(v_i)=|\{(v_i,v_j):(v_i,v_j)\in E_{\text{pot}},v_j\in V_{\text{swp}}\}|
$$
be the out-degree of node $v_i$ in $G_{\text{pot}}$. Recall that except for the auxiliary node $v_0$, $V_{\text{swp}}$ and $V_{\text{pot}}$ share the same vertices $\{v_i\}^t_{i=1}$. In addition, for any $v_i \in V_{\text{swp}}$, let
$$
d(v_i)=|\{\{v_i,v\}:\{v_i,v\}\in E_{\text{swp}},v\in V_{\text{swp}}\}|
$$
denote the number of edges in $E_{\text{swp}}$ that are incident to $v_{i}$. 
In what follows, we interchangeably use $v_i$ to denote a node in $V_{\text{swp}}$ or a node in $V_{\text{pot}}$. We start by proving the following lemma.
\begin{lemma}\label{lem:removeandadd}
 For any node $v_{i} \in V_{\text{swp}}$, $d^{\textup{pot}}_{\textup{in}}(v_{i})\le d(v_{i})$. 
\end{lemma}
\begin{IEEEproof}
Note that every edge in $E_{\text{swp}}$ corresponds to a swap in $I^*$ and every arc in $E_{\text{pot}}$ corresponds to a swap not in $I^*$. 
Suppose to the contrary that $d^{\textup{pot}}_{\textup{in}}> d(v_i)$, and remove all swaps $(i,i+1)\in I^*$ that correspond to edges in 
$$
E^{\text{swp}}(v_i)=\{\{v_i,v\}:\{v_i,v\}\in E_{\text{swp}},v\in V_{\text{swp}}\},
$$
which are edges incident to $v_i$. Then, add swaps that correspond to ingoing arcs in 
$$
E^{\text{pot}}(v_i)=\{(v,v_i):(v,v_i)\in E_{\text{pot}},v\in V_{\text{swp}}\}
$$ 
that emanate from nodes in $V_{\textup{swp}}$ and end in node $v_{i}$. We show next that each added swap contributes $2$ to the total discrepancy. Note that by the definition of an arc in $E_{\text{pot}}$, every swap $(i,i+1)$ that corresponds to an arc $(v_{k_1},v_{i})\in E^{\text{pot}}(v_i)$ increases the discrepancy of the companion set $(S_{2k_1-1}^\prime,S_{2k_1}^\prime)$. It therefore suffices to show that the same swap $(i,i+1)\in E^{\text{pot}}(v_i)$ contributes $+1$ to the set discrepancy of $(S_{2i-1}^\prime,S_{2i}^\prime)$ after removing the swaps corresponding to the edges in $E^{\text{swp}}(v_i)$.  
Otherwise, there would be two arcs $(v_{k_1},v_{i}),(v_{k_2},v_{i})\in E^{\text{pot}}(v_i)$ ($k_1$ and $k_2$ can be equal), such that their corresponding swaps $(i,i+1)$ and $(j,j+1)$ contribute $+1$ and $-1$ to the set discrepancy of $(S_{2i-1}^\prime, S_{2i}^\prime)$, respectively, after swap removal and addition. This implies that the swaps $(i,i+1)$ and $(j,j+1)$ contribute $+1$ and $-1$ to the set discrepancy of $(S_{2i-1}^\prime, S_{2i}^\prime)$, respectively, or vice versa, before swap removal and addition. The swap $(i,i+1)$ or $(j,j+1)$ that increases the set discrepancy of $(S_{2i-1}^\prime, S_{2i}^\prime)$ before swap removal and addition should have been included in $I^*$ (note that this swap does not share an element with the swaps in $I^*$), which contradicts the fact that $(i,i+1),(j,j+1)\notin I^*$. Therefore, each added swap contributes $2$ to the total set discrepancy. Moreover, added swaps do not share an element with each other or with other swaps in $I^*$, because two swaps that share an element cannot both contribute $2$ to the total discrepancy.

Therefore, the added $d^{\textup{pot}}_{\textup{in}}(v_i)$ swaps contribute $2d^{\textup{pot}}_{\textup{in}}(v_i)$ to the total discrepancy, while the removed swaps reduce the total discrepancy by $2 d(v_i)$. Since $d^{\textup{pot}}_{\textup{in}}(v_i) > d(v_i)$, this implies that the total set discrepancy increases after swap removal and addition, which contradicts the assumption of maximality of the total set discrepancy induced by $I^*$. Hence, $d^{\textup{pot}}_{\textup{in}}(v_i)\le d(v_i)$.
\end{IEEEproof}
Define next
\begin{align*}
d^{\text{swp}}(v_i)=&|\{\{v_i,v\}:\{v_i,v\}\in E_{\text{swp}},v\in V_{\text{swp}}\backslash\{v_i\}\}|\\
&+2\cdot\mathbbm{1}(\{v_i,v_i\}\in E_{\text{swp}}), 
\end{align*}
%\textcolor{red}{above, do you really use sets and not ()? Edges/arcs are denoted by ().}
where $\mathbbm{1}(P)=1$ if the event $P$ holds and $\mathbbm{1}(P)=0$ otherwise. Note that each node has at most one self-loop. Hence, we have
\begin{align}\label{eq:dswpe}
\sum^t_{i=1}d^{\text{swp}}(v_i)=2|E_{\text{swp}}|.    
\end{align}
Moreover, by definition of $d(v_i)$ and $d^{\text{swp}}(v_i)$, 
\begin{align}\label{eq:dlessthandswp}
    d(v_i)\le d^{\text{swp}}(v_i).
\end{align}
In the following, we show that 
\begin{align}\label{eq:ddoutgreaterthan3}
d^{\text{swp}}(v_i)+d^{\text{pot}}_{\text{out}}(v_i)\ge 3.
\end{align}
We split the companion sets $(S_{2i-1},$ $S_{2i})$, $i\in[t]$ into three ``types'':
\begin{enumerate}
%    \item Type $1$: $\{a,a+2,a+3,a+5\}$ for some $a\ge 1$, where $\{S_{2i-1},S_{2i}\}=\{\{a,a+5\},\{a+2,a+4\}\}$.
    \item \emph{Type} $1$: $\{a,a+b,a+b+1,a+2b+1\}$ for some $b> 1$, where $\{S_{2i-1},S_{2i}\}=\{\{a,a+2b+1\},\{a+b,a+b+1\}\}$;
    \item \emph{Type} $2$: $\{a,a+b,a+b+c,a+2b+c\}$ for some $b,c>1$, and $a\ge 1$, where $\{S_{2i-1},S_{2i}\}=\{\{a,a+2b+c\},\{a+b,a+b+c\}\}$;
    \item \emph{Type} $3$: $\{a,a+1,a+1+b,a+2+b\}$ for some $a,b\ge 1$, where $\{S_{2i-1},S_{2i}\}=\{\{a,a+2+b\},\{a+1,a+1+b\}\}$.
\end{enumerate}
To show that the four elements $\{\ell_1,\ell_2,\ell_3,\ell_4\}$ $\in$ $S_{2i-1}\cup S_{2i}$, $i\in[t]$, are of one of the three types listed above, we order the four elements as  $\ell_1<\ell_2<\ell_3<\ell_4$, for all $i\in[t]$. Note that by the balancing  property of $(S_{2i-1},S_{2i})$, we have $\ell_2-\ell_1=\ell_4-\ell_3$. Only one of the following can hold: (1) $\ell_2-\ell_1=1$, which belongs to Type $3$; (2) $\ell_2-\ell_1>1$ and $\ell_3-\ell_2=1$, which belongs to Type $1$; (3) $\ell_2-\ell_1>1$ and $\ell_3-\ell_2>1$, which belongs to Type $2$. 

%We note that Type $3$ companion sets can be excluded when $|E^j_{\text{swp}}|+1-|V^j_{\text{swp}}|=0$, i.e., the connected component $G^j_{\text{swp}}$ is acyclic. For any companion sets $(S_{2i-1},S_{2i})$ of Type $3$, it can be verified that at most $2$ swaps from $\{(\ell_1-1,\ell_1),(\ell_2,\ell_2+1),(\ell_3-1,\ell_3),(\ell_4,\ell_4+1)\}$ that share only a single a element with $S_{2i-1}\cup S_{2i}=\{\ell_1,\ell_2,\ell_3,\ell_4\}$ are allowed in $I^*$. This is because no three swaps can each increase the discrepancy of $(S_{2i-1},S_{2i})$ by $1$.
%This is because at most $2$ \textcolor{red}{what is a swap between companion sets? we only have elements in swaps?} swaps between $(S_{2i_1-1},S_{2i_1})$ and another pair of companion sets can contribute $+1$ to the discrepancy of $(S_{2i_1-1},S_{2i_1})$ of Type $3$. \textcolor{red}{- you mean all such sets since you used the word simultaneously?}. 
%Moreover, one can always use either of the swaps $(a,a+1)$ or $(a+b+1,a+2+b)$, in addition to the swaps between companion sets $(S_{2i_1-1},S_{2i_1})$ and other companion sets, to further increase the set discrepancy of $(S_{2i_1-1},S_{2i_1})$ by $2$. This implies that there should be a self-loop for $v_{i_1}$ in $V^j_{\text{swp}}$ if $(S_{2i_1-1},S_{2i_1})$ is of Type $3$, contradicting the assumption that $G^j_{\text{swp}}$ is acyclic.

We characterize in the next proposition some properties of companion sets of Type $1$ and Type $2$. %The proof is delegated to the Appendix.% that will be used repeatedly in the proof.
\begin{lemma}\label{lem:type}
For any $v_i\in V_{\text{swp}}$, %$i\in[t]$, let $(S_{2i-1},S_{2i})$ companion sets of Type $j$, $j\in[2]$. Then 
Equation~\eqref{eq:ddoutgreaterthan3} holds. 
%\begin{align}
%d(v_i)+d_{\text{out}}(v_i)= j+2.
%\end{align}
\end{lemma}
\begin{IEEEproof}
Note that swaps involving elements from Type $1$ companion sets can be grouped into two sets, $\{(a-1,a),(a+b+1,a+b+2),(a+2b,a+2b+1)\}$ and $\{(a,a+1),(a+b-1,a+b),(a+2b+1,a+2b+2)\}$, where swaps in the same group can simultaneously contribute $+1$ to the discrepancy of the companion sets $(S_{2i-1},S_{2i})$, while swaps from different groups cannot simultaneously contribute $+1$ to the set  discrepancy. Hence, $I^*$ cannot include swaps from both groups. When
\begin{align*}
&d^{\text{swp}}(v_i)=|I^*\cap (\{(a-1,a),(a+b+1,a+b+2),\\&(a+2b,a+2b+1)\}\cup\{(a,a+1),(a+b-1,a+b),\\&(a+2b+1,a+2b+2)\} )|=0,    
\end{align*}
i.e., when $v_i$ is an isolated node in $G_{\text{swp}}$, 
one of the groups of swaps corresponds to three outgoing arcs from node $v_i$ in $E_{\text{pot}}$, meaning that $d^{\text{pot}}_{\text{out}}(v_i)=3$. When $d^{\text{swp}}(v_i)\ne 0$, 
without loss of generality, we may assume that 
$$
I^*\cap \{(a,a+1),(a+b-1,a+b),(a+2b+1,a+2b+2)\}= \varnothing.
$$
Then, 
\begin{align*}
    d^{\text{swp}}(v_i)=|I^*\cap \{&(a-1,a),(a+b+1,a+b+2),\\
    &(a+2b,a+2b+1)\}|.
\end{align*}
The swaps in $\{(a-1,a),(a+b+1,a+b+2),(a+2b,a+2b+1)\}\backslash I^*$ correspond to $d^{\text{pot}}_{\text{out}}(v_i)=3-d^{\text{swp}}(v_i)$ outgoing arcs from $v_i$ in $E_{\text{pot}}$. In either case, we have $d^{\text{swp}}(v_i)+d^{\text{pot}}_{\text{out}}(v_i)=3$.

Similarly, we group swaps involving elements from Type $2$ companion sets into two sets, $\{(a-1,a),(a+b,a+b+1),(a+b+c,a+b+c+1),(a+2b+c-1,a+2b+c)\}$ and $\{(a,a+1),(a+b-1,a+b),(a+b+c-1,a+b+c),(a+2b+c-1,a+2b+c)\}$. Similar arguments as the one previously described may be used when $d^{\text{swp}}(v_i)+d^{\text{pot}}_{\text{out}}(v_i)=4$. Consequently, $d^{\text{swp}}(v_i)+d^{\text{pot}}_{\text{out}}(v_i)\ge 3$. 

Finally, we group swaps involving elements from Type $3$ companion sets into two sets, $\{(a,a+1),(a+2+b,a+3+b)\}$ and $\{(a-1,a),(a+1+b,a+2+b)\}$. Note that either $(a,a+1)\in I^*$ or $(a+1+b,a+2+b)\in I^*$, since the two swaps contribute ``oppositely'' to the set discrepancy. Hence, by definition of $d^{\text{swp}}(v_i)$ we have $d^{\text{swp}}(v_i)+d^{\text{pot}}_{\text{out}}(v_i)= 3$. Hence,  ~\eqref{eq:ddoutgreaterthan3} holds.
\end{IEEEproof}
We are now ready to prove ~\eqref{eq:lower}. 
Summing \eqref{eq:ddoutgreaterthan3} over $v_i\in V_{\text{swp}}$, we have
\begin{align}\label{eq:3t}
\sum^t_{i=1}d^{\text{swp}}(v_i)+\sum^t_{i=1}d^{\text{pot}}_{\text{out}}(v_i)\ge 3t
\end{align}
On the other hand, we also have 
\begin{align*}
  \sum^t_{i=0}d^{\text{pot}}_{\text{in}}(v_i)=\sum^t_{i=0}d^{\text{pot}}_{\text{out}}(v_i),  
\end{align*}
Since $d^{\text{pot}}_{\text{in}}(v_0)\le 2$ and $d^{\text{pot}}_{\text{out}}(v_0)=0$, 
we have
\begin{align*}
    \sum^t_{i=1}d^{\text{pot}}_{\text{in}}(v_i)\ge \sum^t_{i=1}d^{\text{pot}}_{\text{out}}(v_i)-2.
\end{align*}
Together with~\eqref{eq:3t}, we arrive at
\begin{align}\label{eq:3t-2}
 \sum^t_{i=1}d^{\text{swp}}(v_i)+\sum^t_{i=1}d^{\text{pot}}_{\text{in}}(v_i)\ge 3t-2. 
\end{align}
By Lemma \ref{lem:removeandadd} and \eqref{eq:dlessthandswp}, we have 
\begin{align}\label{eq:swpgreaterthanpot}
  \sum^t_{i=1}d^{\text{swp}}(v_i)\ge\sum^t_{i=1}d^{\text{pot}}_{\text{in}}(v_i).  
\end{align}
Combining Lemma \ref{lem:optimalswap}, \eqref{eq:dswpe}, \eqref{eq:3t-2}, and \eqref{eq:swpgreaterthanpot}, we establish~\eqref{eq:lower}. 
\section{Lower bound for $p\ge 2$ %\textcolor{red}{Start from here}
} \label{sec:greater_pi}
%\subsection{The proof of Theorem \ref{thm:maingeneralp}}
To prove Theorem \ref{thm:maingeneralp}, 
we consider only swaps $(j,j+p)$ for which $j\in [4t-p]$. One can use the same arguments presented in Section~\ref{sec:lower_bound} to prove a lower bound $D^*(t,p)\ge 3pt-2p(p-1)$ (Theorem \ref{thm:main}) by defining edges and arcs in $G_{\text{swp}}(I_{t,p})$ and $G_{\text{pot}}(I_{t,p})$ based on swaps $(j,j+p)$, $j\in[4t-p]$.
%\subsection{The proof of the lower bound in Theorem \ref{thm:p=2}}

In what follows, we improve this lower bound in Theorem~\ref{thm:maingeneralp} for $p=2$ and prove Theorem \ref{thm:p=2} by considering both swaps $(j,j+1)$, $j\in [4t-1],$ and $(j,j+2)$, $j\in[4t-2]$ in $\cI_{t,2}$. Note that to prove Theorem \ref{thm:p=2}, the arguments for proving Theorem \ref{thm:maingeneralp} and Theorem \ref{thm:main} do not work. This is because Lemma \ref{lem:removeandadd}, which is key to the proofs of Theorem \ref{thm:maingeneralp} and Theorem \ref{thm:main}, does not hold when both swaps $(j,j+1)$, $j\in [4t-1],$ and $(j,j+2)$, $j\in[4t-2]$ in $\cI_{t,2}$ are allowed. Therefore, we need to modify and extend Lemma~\ref{lem:removeandadd} to cases when swaps $(j,j+1)$, $j\in [4t-1],$ and $(j,j+2)$, $j\in[4t-2]$ in $\cI_{t,2}$ are considered. 

Similarly to what we did in the proof of~\eqref{eq:lower} from  Section~\ref{sec:lower_bound}, 
we fix the companion sets $(S_1,\ldots,S_{2t})\in\mathcal{S}_t$ and define an undirected \emph{weighted} graph $G^2_{\text{swp}}(I_{t,2})=(V^2_{\text{swp}}(I_{t,2}),E^2_{\text{swp}}(I_{t,2}))$ and a directed weighted graph $G^2_{\text{pot}}(I_{t,2})=(V^2_{\text{pot}}(I_{t,2}),E^2_{\text{pot}}(I_{t,2}))$ for a set of swaps $I_{t,2}\in\mathcal{I}_{t,2}$. But unlike the \emph{weighted graphs} $G^2_{\text{swp}}(I_{t,2})$ and $G^2_{\text{pot}}(I_{t,2})$, the graphs $G_{\text{swp}}(I_{t,1})$ and $G_{\text{pot}}(I_{t,1})$ from Section~\ref{sec:lower_bound} are \emph{unweighted}. Moreover, as described next, the node sets are different as well.%the node sets are given by $V^2_{\text{swp}}(I_{t,2})=V^2_{\text{swp}}(I_{t,2})=\{v_i\}^{4t}_{i=1}$, where each node $v_i$ represents an element $i\in[4t]$ rather than a pair of companion sets $(S_{2i-1},S_{2i})$

The graph $G^2_{\text{swp}}(I_{t,2})$ is used to describe the swap set $I_{t,2}$. 
The nodes are given by $V^2_{\text{swp}}(I_{t,2})=\{v_j\}^{4t}_{j=1}$ so that  each node $v_j\in V^2_{\text{swp}}(I_{t,2})$ corresponds to an integer $i\in[4t]$. There exists an edge $\{v_{j_1},v_{j_2}\}\in E^2_{\text{swp}}(I_{t,2}))$ between $v_{j_1}$ and $v_{j_2}$ with weight $w_{\text{swp}}(j_1,j_2)>0$ if $(j_1,j_2)\in I_{t,2}$ and $w_{\text{swp}}(j_1,j_2)=|j_1-j_2|$. %The same as in $G_{\text{swp}}(I_{t,1})$, self loops and multiple edges between the same pair of nodes are allowed in $G^2_{\text{swp}}(I_{t,2})$.

The graph $G^2_{\text{pot}}(I_{t,2})$ describes the potential discrepancy changes induced by swaps $(j_1,j_2)\in (\cI_{t,2}\backslash I_{t,2})$. The node set is given by $V^2_{\text{pot}}(I_{t,2})=V^2_{\text{swp}}(I_{t,2})$. The arc set $E^2_{\text{pot}}(I_{t,2})$ is defined as follows: 
an arc $(v_{j},v_{j+2})$ (directed from $v_{j}$ to $v_{j+2}$, $j\in[4t-2]$) of weight $w_{\text{pot}}(j,j+2)=2-\mathbbm{1}(\{j,j+1\}\in E^2_{\text{swp}}(I_{t,2}))$, where $\mathbbm{1}(event)$ equals $1$ if the $event$ is true and $\mathbbm{1}(event)$ equals $0$ otherwise, belongs to the arc set of  $G^2_{\text{pot}}(I_{t,2})$ if $(j,j+2)\notin I_{t,2}$ and there exist two pairs of companion sets $(S_{2i_1-1},S_{2i_1})$ and $(S_{2i_2-1},S_{2i_2})$ ($i_1$ and $i_2$ are allowed to be the same) satisfying one of the following conditions: 
\begin{enumerate}
    \item $j\in S_{2i_1}$, $j+2\in (S_{2i_2-1}\cup S_{2i_2})\backslash S_{2i_1}$, and $\Sigma(S_{2i_1-1}^\prime)<\Sigma(S_{2i_1}^\prime)$, where $S^\prime_i$, $i\in[2t]$ is the set obtained from $S_i$ after performing the swaps in $I_{t,2}$.
    \item $j\in S_{2i_1-1}$, $j+2\in (S_{2i_2-1}\cup S_{2i_2})\backslash S_{2i_1-1}$, and $\Sigma(S_{2i_1-1}^\prime)>\Sigma(S_{2i_1}^\prime)$ after performing the swaps in $I_{t,2}$.
    \item $j\in S_{2i_1-1}$, $j+2\in (S_{2i_2-1}\cup S_{2i_2})\backslash S_{2i_1-1}$, and $\Sigma(S_{2i_1-1}^\prime)=\Sigma(S_{2i_1}^\prime)$ after performing the swaps in $I_{t,2}$.
    %\item $i+p\in S_{2i_1}$, $i\in (S_{2i_2-1}\cup S_{2i_2})\backslash S_{2i_1}$, and $\Sigma(S_{2i_1-1}^\prime)=\Sigma(S_{2i_1}^\prime)$ after performing the swaps in $I_{t,p}$.
\end{enumerate}
An arc $(v_{j},v_{j+1})$ from $v_{j}$ to $v_{j+1}$, $j\in[4t-2]$, of weight $w_{\text{pot}}(j,j+1)=1$ exists if 
$(j,j+1)\notin I_{t,2}$, $j,j+2\in S_{i}$ for some $i\in [2t]$, and there exist two pairs of companion sets $(S_{2i_1-1},S_{2i_1})$ and $(S_{2i_2-1},S_{2i_2})$ (where $i_1$ and $i_2$ can be the same) satisfying one of the following conditions: 
\begin{enumerate}
    \item $j\in S_{2i_1}$, $j+1\in (S_{2i_2-1}\cup S_{2i_2})\backslash S_{2i_1}$, and $\Sigma(S_{2i_1-1}^\prime)<\Sigma(S_{2i_1}^\prime)$, where $S^\prime_i$, $i\in[2t]$ is the set obtained from $S_i$ after performing the swaps in $I_{t,2}$.
    \item $j\in S_{2i_1-1}$, $j+1\in (S_{2i_2-1}\cup S_{2i_2})\backslash S_{2i_1-1}$, and $\Sigma(S_{2i_1-1}^\prime)>\Sigma(S_{2i_1}^\prime)$ after performing the swaps in $I_{t,2}$.
    \item $j\in S_{2i_1-1}$, $j+1\in (S_{2i_2-1}\cup S_{2i_2})\backslash S_{2i_1-1}$, and $\Sigma(S_{2i_1-1}^\prime)=\Sigma(S_{2i_1}^\prime)$ after performing the swaps in $I_{t,2}$.
    %\item $i+p\in S_{2i_1}$, $i\in (S_{2i_2-1}\cup S_{2i_2})\backslash S_{2i_1}$, and $\Sigma(S_{2i_1-1}^\prime)=\Sigma(S_{2i_1}^\prime)$ after performing the swaps in $I_{t,p}$.
\end{enumerate}
Similarly, an arc $(v_{j},v_{j-2})\in E^2_{\text{pot}}(I_{t,2})$, $j\in\{2,\ldots,4t\}$ of weight $w_{\text{pot}}(j,j-2)=2-\mathbbm{1}(\{j-1,j\}\in E^2_{\text{swp}}(I_{t,2}))$ exists if $(j-2,j)\notin I_{t,2}$ and there exist two pairs of companion sets $(S_{2i_1-1},S_{2i_1})$ and $(S_{2i_2-1},S_{2i_2})$ (where $i_1$ and $i_2$ are allowed to be the same) satisfying one of the following conditions: 
\begin{enumerate}
    \item $j\in S_{2i_1}$, $j-2\in (S_{2i_2-1}\cup S_{2i_2})\backslash S_{2i_1}$, and $\Sigma(S_{2i_1-1}^\prime)>\Sigma(S_{2i_1}^\prime)$ after performing the swaps in $I_{t,2}$.
    \item $j\in S_{2i_1-1}$, $j-2\in (S_{2i_2-1}\cup S_{2i_2})\backslash S_{2i_1-1}$, and $\Sigma(S_{2i_1-1}^\prime)<\Sigma(S_{2i_1}^\prime)$ after performing the swaps in $I_{t,2}$.
    \item $j\in S_{2i_1}$, $j-2\in (S_{2i_2-1}\cup S_{2i_2})\backslash S_{2i_1}$, and $\Sigma(S_{2i_1-1}^\prime)=\Sigma(S_{2i_1}^\prime)$ after performing the swaps in $I_{t,2}$.
    %\item $i+p\in S_{2i_1}$, $i\in (S_{2i_2-1}\cup S_{2i_2})\backslash S_{2i_1}$, and $\Sigma(S_{2i_1-1}^\prime)=\Sigma(S_{2i_1}^\prime)$ after performing the swaps in $I_{t,p}$.
\end{enumerate}
An arc $(v_{j},v_{j-1})\in E^2_{\text{pot}}(I_{t,2})$, $j\in\{3,\ldots,4t\}$, of weight $w_{\text{pot}}(j,j-1)=1$ is present if $(j-1,j)\notin I_{t,2}$, $j,j-2\in S_{i}$ for some $i\in [2t]$, and there exist two pairs of companion sets $(S_{2i_1-1},S_{2i_1})$ and $(S_{2i_2-1},S_{2i_2})$ (where $i_1$ and $i_2$ are allowed to be the same) satisfying one of the following conditions: 
\begin{enumerate}
    \item $j\in S_{2i_1}$, $j-1\in (S_{2i_2-1}\cup S_{2i_2})\backslash S_{2i_1}$, and $\Sigma(S_{2i_1-1}^\prime)>\Sigma(S_{2i_1}^\prime)$ after performing the swaps in $I_{t,2}$.
    \item $j\in S_{2i_1-1}$, $j-1\in (S_{2i_2-1}\cup S_{2i_2})\backslash S_{2i_1-1}$, and $\Sigma(S_{2i_1-1}^\prime)<\Sigma(S_{2i_1}^\prime)$ after performing the swaps in $I_{t,2}$.
    \item $j\in S_{2i_1}$, $j-1\in (S_{2i_2-1}\cup S_{2i_2})\backslash S_{2i_1}$, and $\Sigma(S_{2i_1-1}^\prime)=\Sigma(S_{2i_1}^\prime)$ after performing the swaps in $I_{t,2}$.
    %\item $i+p\in S_{2i_1}$, $i\in (S_{2i_2-1}\cup S_{2i_2})\backslash S_{2i_1}$, and $\Sigma(S_{2i_1-1}^\prime)=\Sigma(S_{2i_1}^\prime)$ after performing the swaps in $I_{t,p}$.
\end{enumerate} 
For convenience, we assume that $w_{\text{pot}}(j_1,j_2)=0$ and $w_{\text{swp}}(j_1,j_2)=0$ if $(j_1,j_2)\notin E^2_{\text{pot}}(I_{t,2})$ and $\{j_1,j_2\}\notin E^2_{\text{swp}}(I_{t,2})$, respectively. 
Intuitively, an arc $(v_{j_1},v_{j_2})\in E^2_{\text{pot}}(I_{t,2})$ indicates that adding the swap $(j_1,j_2)$ and removing one of the swaps $(j_1,j_1+1)$ or $(j_1-1,j_1)$ in $I_{t,2}$ associated with $j_1$ maximally increases the discrepancy $|\Sigma (S^\prime_{2i_1-1})-\Sigma (S^\prime_{2i_1})|$ of the companion sets $(S^\prime_{2i_1,-1},S^\prime_{2i_1}),$ where $j_1\in (S^\prime_{2i_1,-1}\cup S^\prime_{2i_1})$. In addition, the weight $w_{\text{pot}}(j_1,j_2)$ reflects the increase in this discrepancy. 

Next, fix the defining sets $(S_1,\ldots,S_{2t})$, and let $I^*_{t,2}$ be a swap set that has the smallest size and results in a maximal total discrepancy. Then, similar to Lemma \ref{lem:optimalswap}, we have the following result.
\begin{lemma}\label{lem:optimaldiscrepancyp2}
The total discrepancy of the defining sets $(S_1,\ldots,S_{2t})$ under swaps in $I^*_{t,2}$ is given by
\begin{align}\label{eq:optimaldiscrepancyp2}
D(S_1,\ldots,S_{2t};I^*_{t,2};t)=\sum_{j_1\in[4t]}\sum_{j_2\in[4t]}w_{\text{swp}}(j_1,j_2).  
\end{align}
\end{lemma}
\begin{proof} 
Similarly to the proof of Lemma ~\ref{lem:optimalswap}, we will show that every swap $(j_1,j_2)\in I^*_{t,2}$ contributes $w_{\text{swp}}(j_1,j_2)=|j_1-j_2|$ to the discrepancy of each of the companion sets $(S_{2i_1-1},S_{2i_1})$ and $(S_{2i_2-1},S_{2i_2})$, respectively, where $j_1\in (S_{2i_1-1}\cup S_{2i_1})$ and $j_2\in (S_{2i_2-1}\cup S_{2i_2})$ (and $i_1$ and $i_2$ are allowed to be the same). Suppose, on the contrary, that there exists a swap $(j^*_1,j^*_2)\in I^*_{t,2}$ that contributes less than $w_{\text{swp}}(j_1,j_2)$ to the discrepancy of each of the companion sets $(S_{2i^*_1-1},S_{2i^*_1}),$ where $j^*_1\in (S_{2i^*_1-1}\cup S_{2i^*_1})$ or $j^*_2\in (S_{2i^*_1-1}\cup S_{2i^*_1})$. Then, there exists a swap $(j^*_3,j^*_4)\in I^*_{t,2}$ that contributes $-w_{\text{swp}}(j*_3,j^*_4)=-|j^*_3-j^*_4|$ to the discrepancy of $(S_{2i^*_1-1},S_{2i^*_1})$. Hence, removing $(j^*_3,j^*_4)$ from $I^*_{t,2}$ does not decrease the total discrepancy since $(j^*_3,j^*_4)$ contributes a discrepancy of at most $w_{\text{swp}}(j*_3,j^*_4)$ to another pair of companion sets. This contradicts the minimality assumption on  the size of $I^*_{t,2}$.
\end{proof}
Before proceeding to the proof of the lower bound, we present a few lemmas. 
The first lemma is a generalization of Lemma ~\ref{lem:removeandadd}.
\begin{lemma}\label{lem:optlessswpbi}
For an edge and an arc set $E^2_{\text{swp}}(I^*_{t,2})$ and $E^2_{\text{pot}}(I^*_{t,2}),$ respectively, and 
for any pair of companion sets $(S_{2i-1},S_{2i})$, $i\in[t]$, we have
\begin{align}\label{eq:potlessthandp2}
    &\sum_{j_1\in (S_{2i-1}\cup S_{2i})}\sum_{j_2\in[4t]}w_{\text{opt}}(j_2,j_1)\nonumber\\
\le&    \sum_{j_1\in (S_{2i-1}\cup S_{2i})}\sum_{j_2\in[4t]}w_{\text{swp}}(j_2,j_1)+b_i
\end{align}
where $$b_i=\sum_{j_1\in (S_{2i-1}\cup S_{2i})}\big(\sum_{j_2\in[4t]}w_{\text{pot}}(j_2,j_1)-\max_{j_2\in[4t]}w_{\text{pot}}(j_2,j_1)\big).$$  
%where $b_i$ is the number of ingoing arcs that are of the form $(v_{j},v_{j-1})$ or $(v_{j-1},v_{j})$ to nodes in $\{v_j:j\in (S_{2i-1}\cup S_{2i})\}$, i.e., $$b_i=|\{(v_{j_1},v_{j_2}):j_2\in (S_{2i-1}\cup S_{2i}),|j_1-j_2|=1\}|.$$ 
\end{lemma}
\begin{proof}
%We first note that for each node $v_j$, there is at most one ingoing arc of the form $(v_{j-2},v_j)$ or $(v_{j+2},v_j)$ to $v_j$. Otherwise, if $(v_{j-2},v_j),(v_{j},v_{j+2})\in E^2_{\text{pot}}(I^*_{t,2})$, then we have either $\{v_{j-2},v_j\}\in E^2_{\text{swp}}(I^*_{t,2})$ or $\{v_{j+2},v_j\}\in E^2_{\text{swp}}(I^*_{t,2})$ since exactly one swap of $(j-2,j)$ and $(j+2,j)$ increases the set discrepancy of $(S_{2i-1},S_{2i})$. On the other hand, either $\{v_{j-2},v_j\}\in E^2_{\text{swp}}(I^*_{t,2})$ or $\{v_{j+2},v_j\}\in E^2_{\text{swp}}(I^*_{t,2})$ contradicts the assumption $(v_{j-2},v_j),(v_{j},v_{j+2})\in E^2_{\text{pot}}(I^*_{t,2})$. Hence, there is at most one of the arcs $(v_{j-2},v_j)$ and $(v_{j+2},v_j)$ that directs to $v_j$. The rest of 
The proof follows along the same lines as that of Lemma ~\ref{lem:removeandadd}. We have that 
\begin{align}\label{eq:potlessswp1}
&\sum_{j_1\in (S_{2i-1}\cup S_{2i})}\max_{j_2\in[4t]}w^2_{\text{pot}}(j_2,j_1)\nonumber\\
\le& \sum_{j_1\in (S_{2i-1}\cup S_{2i})}\sum_{j_2\in[4t]}w^2_{\text{swp}}(j_2,j_1), \end{align}
which holds because one can otherwise remove the swaps associated with nodes $\{v_j:j\in(S_{2i-1}\cup S_{2i})\}$ and add one swap $\{j_2,j_1\}$ for each $j_1\in (S_{2i-1}\cup S_{2i})$ such that $w^2_{\text{pot}}(j_2,j_1)=\max_{j\in [4t]}w^2_{\text{pot}}(j,j_1)$, which leads to a higher set discrepancy. Note that each node $v_j$, $j\in[4t]$, has at most one outgoing arc and thus, such a swap replacement is valid. Then, ~\eqref{eq:potlessthandp2} follows from \eqref{eq:potlessswp1}.
%\begin{align*}
%\{\{j_2,j_1\}:&(j-2,j)\in E^2_{\text{pot}}(I^*_{t,2})  \text{ and }j\in (S_{2i-1}\cup S_{2i})\\
%\text{or }&(j,j-2)\in E^2_{\text{pot}}(I^*_{t,2})\text{ and }j-2\in(S_{2i-1}\cup S_{2i})\} 
%\end{align*}
%to obtain a higher set discrepancy. Combining \eqref{eq:potlessswp} and the fact that each node $v_j$ is associated with at most one of $(j-2,j)$ and $(j+2,j)$ in $E^2_{\text{pot}}(I^*_{t,2})$, we have \eqref{eq:potlessthandp2}.
\end{proof}
Next, we provide bounds on $\sum_{i\in[t]}b_i$. We first show that for each node $v_j$, there is at most one ingoing arc of the form $(v_{j-2},v_j)$ or $(v_{j+2},v_j)$ to $v_j$. Otherwise, if $(v_{j-2},v_j),(v_{j},v_{j+2})\in E^2_{\text{pot}}(I^*_{t,2})$, then we have either $\{v_{j-2},v_j\}\in E^2_{\text{swp}}(I^*_{t,2})$ or $\{v_{j+2},v_j\}\in E^2_{\text{swp}}(I^*_{t,2})$ since exactly one swap of the form $(j-2,j)$ and $(j+2,j)$ increases the set discrepancy of $(S_{2i-1},S_{2i})$. On the other hand, either $\{v_{j-2},v_j\}\in E^2_{\text{swp}}(I^*_{t,2})$ or $\{v_{j+2},v_j\}\in E^2_{\text{swp}}(I^*_{t,2})$ contradicts the assumption $(v_{j-2},v_j),(v_{j},v_{j+2})\in E^2_{\text{pot}}(I^*_{t,2})$. Hence, at most one of the arcs $(v_{j-2},v_j)$ and $(v_{j+2},v_j)$ is directed towards $v_j$. Moreover, note that at most one of $(v_{j-1},v_j)$ and $(v_{j+1},v_j)$ exists (because either arc implies $\{j-1,j+1\}=S_{i}$ for some $i\in[2t],$ and cannot simultaneously satisfy both $(v_{j-1},v_j)\in E^2_{\text{pot}}(I^*_{t,2})$ and $(v_{j+1},v_j)\in E^2_{\text{pot}}(I^*_{t,2})$). 
This implies that any node $v_{j_1}$, $j_1\in(S_{2i-1}\cup S_{2i})$ contributes $1$ to the value of $b_i$ only when one of $(v_{j_1-2},v_{j_1})$ and $(v_{j_1+2},v_{j_1})$ exists in $E^2_{\text{pot}}(I_{t,2})$ and one of $(v_{j_1-1},v_{j_1})$ and $(v_{j_1+1},v_{j_1})$ exists in $E^2_{\text{pot}}(I_{t,2})$, in which case, we have $\{j_1-1,j_1+1\}=S_i$ for some $i\in[2t]$.

Consider the following sets of companion sets:
\begin{align}\label{eq:defa1}
    A_1=&\{(S_{2i-1},S_{2i}):j,j+2\in S_{2i-1} \textup{ or }j,j+2\in S_{2i}, \nonumber\\
    &\textup{ for some integer }j\in[4t-2]\};\nonumber\\
    A_2=&\{(S_{2i-1},S_{2i}):j,j+2\in S_{2i-1} \textup{ or }j,j+2\in S_{2i}, \nonumber\\
    &\textup{ for some integer }j\in[4t-2],\textup{ and}\nonumber\\
    &(j,j+1),(j-1,j+1)\in E^2_{\text{pot}}(I_{t,2}) \textup{ or}\nonumber\\
    &(j+2,j+1),(j+3,j+1)\in E^2_{\text{pot}}(I_{t,2})\};\nonumber\\
    A_3=&\{(S_{2i-1},S_{2i}):j,j+2\in S_{2i-1} \textup{ or }j,j+2\in S_{2i}, \nonumber\\
    &\textup{ for some integer }j\in[4t-2],\textup{ and}\nonumber\\
    &(j,j+1),(j+3,j+1)\in E^2_{\text{pot}}(I_{t,2}) \textup{ or}\nonumber\\
    &(j+2,j+1),(j-1,j+1)\in E^2_{\text{pot}}(I_{t,2})\}.   
    % \Big\{(S_{2i-1},S_{2i}):(S_{2i-1},S_{2i})\in A_1, \textup{ and }\nonumber\\
    % &\sum_{j_1\in (S_{2i-1}\cup S_{2i})}\sum_{j_2\in[4t]}w_{\text{opt}}(j_1,j_2)\nonumber\\
    % <&\sum_{j_1\in (S_{2i-1}\cup S_{2i})}\sum_{j_2\in[4t]}w_{\text{swp}}(j_1,j_2)\Big\}
\end{align}
Then, each node $v_j$, $j\in(S_{2i-1}\cup S_{2i})$ that contributes $1$ to the value of $b_i$ corresponds to a unique companion set $(S_{2i-1},S_{2i})\in A_2$ or $(S_{2i-1},S_{2i})\in A_3$. Hence, we have
\begin{align}\label{eq:sumofbi}
\sum_{i\in[t]}b_i=|A_2|+|A_3|
\end{align}
The following proposition improves the bound ~\eqref{eq:potlessthandp2} for arbitrary companion sets $(S_{2i-1},S_{2i})\in A_2$.
\begin{proposition}\label{prop:potlessswp} 
For any set $(S_{2i-1},S_{2i})\in A_2$, we have 
\begin{align}\label{eq:potlessswp}
    &\sum_{j_1\in (S_{2i-1}\cup S_{2i})}\sum_{j_2\in[4t]}w_{\text{pot}}(j_1,j_2)\nonumber\\
     <&\sum_{j_1\in (S_{2i-1}\cup S_{2i})}\sum_{j_2\in[4t]}w_{\text{swp}}(j_1,j_2)+b_i.
\end{align}    
\end{proposition}
\begin{proof}
The proof follows similar to that of Lemma \ref{lem:optlessswpbi}. Without loss of generality, assume that $j,j+2\in S_{2i-1}$ and $(j,j+1),(j-1,j+1)\in E^2_{\text{pot}}(I^*_{t,2})$. Then, there is no edge incident to $j$ in $E^2_{\text{swp}}(I^*_{t,2})$. If \eqref{eq:potlessswp} does not hold, one can remove the swaps associated with nodes $\{v_j:j\in(S_{2i-1}\cup S_{2i})\}$ and add one swap $\{j_2,j_1\}$ for each $j_1\in (S_{2i-1}\cup S_{2i})$ such that $w^2_{\text{pot}}(j_2,j_1)=\max_{j\in [4t]}w^2_{\text{pot}}(j,j_1)$ (note that there is no added swap including $j$). In addition, one can add the swap $\{j-1,j\}$. The added swaps result in higher discrepancy, which contradicts the definition of $I^*_{t,2}$. Hence,~\eqref{eq:potlessswp} holds.    
\end{proof}
The following is an extension of Lemma \ref{lem:type} to the case $p=2$.
\begin{lemma}
For the graphs $G^2_{\text{pot}}(I^*_{t,2})$ and $G^2_{\text{swp}}$, we have
\begin{align}\label{eq:sumpotswp}
    &\sum_{j_1\in [4t]}\sum_{j_2\in[4t]}w_{\text{pot}}(j_2,j_1)\nonumber\\
&+    \sum_{j_1\in [4t]}\sum_{j_2\in[4t]}w_{\text{swp}}(j_2,j_1)\nonumber\\
\ge&8t-|A_1|-8,
\end{align}
where $A_1$ is defined in~\eqref{eq:defa1}.
\end{lemma}
\begin{proof}
Note that for any node $v_j$ such that $j\in (S_{2i-1}\cup S_{2i})$ for some $(S_{2i-1},S_{2i})\notin A_1$ and $j\notin \{1,2,4t-1,4t\}$, we have
\begin{align}\label{eq:singlenode}
  \sum_{j'\in[4t]} w^2_{\text{pot}}(j,j')+\sum_{j'\in[4t]} w^2_{\text{swp}}(j,j') =2.
\end{align}
Moreover, for any $(S_{2i-1},S_{2i})\in A_1$, there is at most one $j\in (S_{2i-1}\cup S_{2i})$ such that
\begin{align}\label{eq:singlenode1}
  \sum_{j'\in[4t]} w^2_{\text{pot}}(j,j')+\sum_{j'\in[4t]} w^2_{\text{swp}}(j,j') =1,
\end{align}
and for the remaining $j'\in (S_{2i-1}\cup S_{2i})\backslash\{j,1,2,4t-1,4t\}$, we have~\eqref{eq:singlenode}. Combining \eqref{eq:singlenode} and \eqref{eq:singlenode1}, we obtain
\begin{align*}
    &\sum_{j_1\in [4t]}\sum_{j_2\in[4t]}w_{\text{pot}}(j_1,j_2)\nonumber\\
&+    \sum_{j_1\in [4t]}\sum_{j_2\in[4t]}w_{\text{swp}}(j_1,j_2)\nonumber\\
\ge&8t-|A_1|-8.   
\end{align*}
On the other hand, we have 
\begin{align*}
    &\sum_{j_1\in (S_{2i-1}\cup S_{2i})}\sum_{j_2\in[4t]}w_{\text{pot}}(j_1,j_2)\nonumber\\
=&    \sum_{j_1\in (S_{2i-1}\cup S_{2i})}\sum_{j_2\in[4t]}w_{\text{pot}}(j_2,j_1).
\end{align*}
Therefore,~\eqref{eq:sumpotswp} holds.
\end{proof}
According to Lemma~\ref{lem:optlessswpbi} and Proposition~\ref{prop:potlessswp}, we have
\begin{align}\label{eq:optlessswpa2}
    &\sum_{j_1\in [4t]}\sum_{j_2\in[4t]}w_{\text{opt}}(j_2,j_1)\nonumber\\
\le&    \sum_{j_1\in [4t]}\sum_{j_2\in[4t]}w_{\text{swp}}(j_2,j_1)+\sum_{i\in[t]}b_i-|A_2|.   
\end{align}
%On the other hand, we have
%\begin{align}\label{eq:bia1}
%   \sum_{i\in[t]}b_i\le |A_1|, 
%\end{align}
%since any arc $(v_{j_1},v_{j_2})$ with $|j_1-j_2|=1$ comes from a node in $(S_{2i-1}\cup S_{2i})$ for some $(S_{2i-1},S_{2i})\in A_1$.
Combining \eqref{eq:sumofbi},  \eqref{eq:sumpotswp}, and \eqref{eq:optlessswpa2}, we arrive at
\begin{align}\label{eq:swplessthan}
    &\sum_{j_1\in [4t]}\sum_{j_2\in[4t]}w_{\text{swp}}(j_2,j_1)
\ge&\frac{8t-|A_1|-|A_3|-8}{2}.%\nonumber\\
%\ge&\frac{8t-|A_2|-2|A_3|-8}{2}
\end{align}
The following proposition provides another lower bound for $\sum_{j_1\in [4t]}\sum_{j_2\in[4t]}w_{\text{swp}}(j_2,j_1)$.
\begin{proposition}
In the graph $G^2_{\text{swp}}(I^*_{t,2})$, we have
 \begin{align}\label{eq:swpmorethan}
     &\sum_{j_1\in [4t]}\sum_{j_2\in[4t]}w_{\text{swp}}(j_2,j_1)\nonumber\\
\ge&5|A_3|.  
\end{align}   
\end{proposition}
\begin{proof}
We first show that for any %$(S_{2i-1},S_{2i})\in A_2$ or 
$(S_{2i-1},S_{2i})\in A_3$, where $S_{2i-1}=\{j,j+2\}$ or $S_{2i}=\{j,j+2\}$, we either have $\{j-1,j+1\}\in E^2_{\text{swp}}(I^*_{t,2})$ or $\{j+3,j+1\}\in E^2_{\text{swp}}(I^*_{t,2})$. 
%Let $(S_{2i-1},S_{2i})\in A_3$ be 
Without loss of generality assume that 
$(j,j+1),(j+3,j+1)\in E^2_{\text{pot}}(I^*_{t,2})$. Then if $\{j-1,j+1\}\notin E^2_{\text{swp}}(I^*_{t,2})$, we have either $\{j,j+1\}\in E^2_{\text{swp}}(I^*_{t,2})$ or $\{j+3,j+1\}\in E^2_{\text{swp}}(I^*_{t,2})$ because $\{j+2,j+1\}\notin E^2_{\text{swp}}(I^*_{t,2})$. This contradicts the assumption that $(j,j+1),(j+3,j+1)\in E^2_{\text{pot}}(I^*_{t,2})$. Hence, $\{j-1,j+1\}\in E^2_{\text{swp}}(I^*_{t,2})$.  Note that $\{j-1,j+1\}\in E^2_{\text{swp}}(I^*_{t,2})$ %or $\{j+3,j+1\}\in E^2_{\text{swp}}(I^*_{t,2})$ 
is uniquely associated with a pair of companion sets $(S_{2i-1},S_{2i})$. In addition, since $(j,j+1),(j+3,j+1)\in E^2_{\text{pot}}(I^*_{t,2})$ and $\{j,j+2\}=S_{2i-1}$ or $\{j,j+2\}=S_{2i}$, we have either $\{j+2,j+3\}\in E^2_{\text{swp}}(I^*_{t,2})$ or $\{j+2,j+4\}\in E^2_{\text{swp}}(I^*_{t,2})$, which is associated with at most two pairs of companion sets $(S_{2i_1-1},S_{2i_1})$ and $(S_{2i_2-1},S_{2i_2})$ in $A_3$. Hence, the total weight of the edges in $E^2_{\text{swp}}(I^*_{t,2})$ associated with companion sets in $A_3$ is at least $2.5|A_3|$, which implies~\eqref{eq:swpmorethan}. %which results in at least $5|A_3|$ total discrepancy according to Lemma \ref{lem:optimaldiscrepancyp2}.
%Similarly, if $(j+2,j+1),(j-1,j+1)\in E^2_{\text{pot}}(I^*_{t,2})$, we have $\{j+3,j+1\}\in E^2_{\text{swp}}(I^*_{t,2})$.
\end{proof}
From \eqref{eq:swplessthan} and \eqref{eq:swpmorethan}, we have that
\begin{align*}
    \sum_{j_1\in [4t]}\sum_{j_2\in[4t]}w_{\text{swp}}(j_2,j_1)\ge&\frac{10}{11}\frac{8t-|A_1|-|A_3|-8}{2}+\frac{|A_3|}{11}\\
    =&\frac{5(8t-|A_1|-8)}{11}\\
    \ge &\frac{35t-40}{11},
\end{align*}
which together with Lemma~\ref{lem:optimaldiscrepancyp2} proves the lower bound in Theorem~\ref{thm:p=2}.

\section{An Upper Bound for $p=2$} \label{section:pitwo}

We provide next a construction of companion sets $\{(S_{2i-1},S_{2i})\}_{i\in[t]}$ that achieves $D(S_1,\ldots,S_{2t};I_{t,2};t)=\frac{9t-1}{2}$.  
For $t=2z+1$, $z\ge 1$, consider the following companion sets:
\begin{align}\label{eq:compsestsp2}
    &S_{1}=\{1,5\},~S_2= \{2,4\}\nonumber\\
    &S_{2i-1}=\{3+8(\ell-1),10+8(\ell-1)\},\text{ and}\nonumber\\
    &S_{2i}=\{6+8(\ell-1),7+8(\ell-1)\},~i=2\ell,\ell\in [z]\nonumber\\
    &S_{2i-1}=\{8+8(\ell-1),13+8(\ell-1)\},\text{ and}\nonumber\\
    &S_{2i}=\{9+8(\ell-1),12+8(\ell-1)\},~i=2\ell+1,\ell\in[z-1]\nonumber\\
    &S_{4z+1}=\{8z,8z+4\},~S_{4z+2}= \{8z+1,8z+3\}.
\end{align}
\begin{figure}[t]
\centering
\includegraphics[width=1\linewidth]{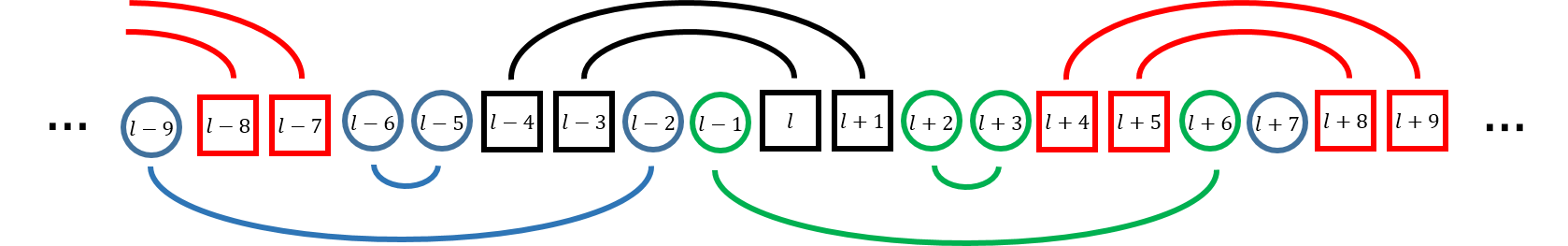}
\caption{An illustration of the recursive partition from~\eqref{eq:compsestsp2}. Integers within the same companion sets have the same color.}
\label{fig:partitionp=2}
\end{figure}
The sets are constructed in a recursive manner (except for $(S_1,S_2)$ and $(S_{4z+1},S_{4z+2})$ that do not follow the same pattern as $(S_3,\ldots,S_{4z}$), and are illustrated in Figure~\ref{fig:partitionp=2}). 
Before presenting the proof of the upper bound, we compare our construction and bounds  with numerical results taken from our earlier work~\cite{pan2022balanced} on the number of optimal defining sets and the corresponding set discrepancy for $p=2$.
\begin{table}[h!]
\caption{Number of optimal partitions and the corresponding set discrepancy for $p=2$}
\label{tab:p=2}
\centering
\begin{tabular}{ |c|c|c|}
 \hline
 Value of $t$& Number of optimal partitions &Optimal set discrepancy\\
 \hline 
 3& 1 &12\\
 \hline
 4& 12 &18\\
 \hline
 5& 7 &22\\
 \hline
\end{tabular}
\end{table}
One can check that our upper  bound is optimal when $t=5$. In addition, our construction \eqref{eq:compsestsp2} is the unique optimal solution for the case when $t=3$.

In what follows, we give an upper bound on the maximum set discrepancy caused by swaps. 
Let $I^*_{z}$ be the set of swaps that causes the maximum total discrepancy in $(S_1,\ldots,S_{4z+2})$ using a minimum number of swaps. Then by similar arguments to those from the proof of Lemma~\ref{lem:optimaldiscrepancyp2}, it can be shown that each swap $(j_1,j_2)\in I^*_z$ contributes $1$ or $2$ to a pair of companion sets $(S_{2i-1},S_{2i})$ where $j_1\in (S_{2i-1}\cup S_{2i})$ or $j_2\in (S_{2i-1}\cup S_{2i})$. 
We use recursion to obtain the upper bound. Note that one of $(1,3)$, $(2,3)$, $(3,4)$, and $(3,5)$ is in $I^*_{S}$ since $3$ can be swapped with one of $1,2,4$, and $5$ to increase the set discrepancy of $(S_3,S_4)$. 
Let $d^2_{z,l}$ denote the maximum discrepancy caused by $I^*_z$ when $(1,3)\in I^*_z$ or $(2,3)\in I^*_z$. Let $d^2_{z,r}$ denote the maximum discrepancy caused by $I^*_z$ when $(3,4)\in I^*_z$ or $(3,5)\in I^*_z$. Then, we have that
\begin{align}\label{eq:d2zrecursion}
    &d^2_{z,l}\le \max\{d^2_{z-1,r}+10,d^2_{z-1,l}+8\}\nonumber\\
    &d^2_{z,r}\le \max\{d^2_{z-1,r}+8,d^2_{z-1,l}+8\},
\end{align}
which implies that
\begin{align}\label{eq:zzminus2}
    \max\{d^2_{z,l},d^2_{z,r}\}\le \max\{d^2_{z-2,r}+18,d^2_{z-2,l}+16\}.  
\end{align}
It can be verified that $d^2_{1,l}=12$ and $d^2_{1,r}=12$ (Table~\ref{tab:p=2}). Therefore, we have that the maximum discrepancy in $(S_1,\ldots,S_{4z+2})$ is given by $\max\{d^2_{z,l},d^2_{z,r}\}\le 9(z-1)+12=9z+3,$ for odd $z,$ and by $\max\{d^2_{z,l},d^2_{z,r}\}\le 9(z-2)+22=9z+4,$ for even $z$. Therefore, the upper bound~\eqref{eq:upperp=2} on $D^*(t,2)$ holds.

\section{Bounds for $|S_i|> 2$ and $p=1$} \label{section:greater_defining}

In this section, we consider the case when the sizes of companion sets $S_i$, $i\in[2t]$, are equal to $q,$ and $q>2$. Note that in this case, $(S_1,\ldots,S_{2t})$ is a partition of $[2tq]$. We show that similar bounds to those in Theorem \ref{thm:main} hold. For the lower bound, the arguments closely resemble those from Section~\ref{sec:lower_bound}. We define graphs $G_{\text{swp}}(I_{t,1})$ and $G_{\text{pot}}(I_{t,1})$ and let $I^*_{t,1}$ be the swap set that causes the largest total discrepancy in $(S_1,\ldots,S_{2t})$ using the minimum number of swaps. Then for $G_{\text{swp}}(I^*)$ and $G_{\text{pot}}(I^*)$, Lemma \ref{lem:optimalswap} and Lemma \ref{lem:removeandadd} hold. Next, we show that Lemma~\ref{lem:type} holds as well. Note that for any pair of companion sets $(S_{2i-1},S_{2i})$ satisfying $\sum(S_{2i-1})=\sum(S_{2i})$, we can find integers $j_1<j_2<\ldots<j_6$ such that $\{j_1,j_1+1,\ldots,j_2\}\subseteq S_{2i-1}$, $\{j_3,j_3+1,\ldots,j_4\}\subseteq S_{2i}$, $\{j_5,j_5+1,\ldots,j_6\}\subseteq S_{2i-1}$ or $\{j_1,j_1+1,\ldots,j_2\}\subseteq S_{2i}$, $\{j_3,j_3+1,\ldots,j_4\}\subseteq S_{2i-1}$, $\{j_5,j_5+1,\ldots,j_6\}\subseteq S_{2i}$. Without loss of generality, suppose that $\{j_1,j_1+1,\ldots,j_2\}\subseteq S_{2i-1}$, $\{j_3,j_3+1,\ldots,j_4\}\subseteq S_{2i}$, $\{j_5,j_5+1,\ldots,j_6\}\subseteq S_{2i-1}$. Then, according to a similar argument to that presented in the proof of Lemma~\ref{lem:type}, we consider two groups of swaps or potential swaps $\{(j_1-1,j_1),(j_4,j_4+1),(j_5-1,j_5)\}$ and $\{(j_2,j_2+1),(j_3-1,j_3),(j_6,j_6+1)\}$, where the swaps in each group can be the same (e.g., $j_4=j_5-1$). Then, we have~\eqref{eq:ddoutgreaterthan3}, and thus Lemma~\ref{lem:type} holds. Matching the remaining arguments in Section~\ref{sec:lower_bound}, we arrive at~\eqref{eq:3t-2} from Lemma~\ref{lem:optimalswap}, Lemma~\ref{lem:type}, and Lemma~\ref{lem:removeandadd}. The lower bound~\eqref{eq:lower} consequently follows from~\eqref{eq:3t-2} and~\eqref{eq:swpgreaterthanpot}.

Next, we show that when $q=2k$ is even, a modification of the construction~\eqref{eq:construction} achieves a worst case total discrepancy upper-bounded by $\frac{8t-2}{5}$. The idea is to expand each integer in $S_i$, $i\in[2t],$ into $k$  consecutive integers. Let $t=5\cdot 2^{z-1}-1$ and $(S^{z+1}_{1},\ldots,S^{z+1}_{5\cdot 2^z-2})$ be the sets defined in~\eqref{eq:construction}. Let
$S^{z+1,q}_{i}=\{j:\lceil \frac{j}{q}\rceil\in S^{z+1}_i\}$ for $i\in [2t]=[5\cdot 2^z-2]$. Let $I^*_{t,z+1}$ be the set of swaps of minimum size that causes the maximum total discrepancy in $(S_1,\ldots,S_{2t})$. Note that $(j,j+1)\notin I^*_{t,z+1}$ for $j\not\equiv 0\text{ or }-1 \bmod k$. Moreover, we create a swap set $I_{t,1}$ from $I^*_{t,z+1}$ by adding $(\lceil\frac{j}{k}\rceil,\lceil\frac{j+1}{k}\rceil)$ to $I_{t,1}$ when $(j,j+1)\in I^*_{t,z+1}$. Hence, each swap $(j,j+1)\in I^*_{t,z+1}$ corresponds to a unique swap $(\lceil\frac{j}{k}\rceil,\lceil\frac{j+1}{k}\rceil)$ to $I_{t,1}$ in $I_{t,1}$ and each swap $(\lceil\frac{j}{k}\rceil,\lceil\frac{j+1}{k}\rceil)$ in $I_{t,1}$ contributes $2$ to the total set discrepancy. By Theorem~\ref{thm:main}, we have that $2|I_{t,1}|\le \frac{8t-2}{5}$. Therefore, the total discrepancy caused by $I^*_{t,z+1}$ is at most $2|I^*_{t,z+1}|=2|I_{t,1}|\le \frac{8t-2}{5}$.

\bibliographystyle{IEEEtran}
\bibliography{biblio1.bib}

\end{document}